\RequirePackage{fix-cm}
\documentclass[smallextended]{svjour3}       
\smartqed  
\usepackage{graphicx}
\usepackage{url}
\usepackage{graphicx}
\usepackage{subcaption}
\usepackage{adjustbox}
\usepackage[linesnumbered,ruled,noend]{algorithm2e}
\usepackage{multirow}
\usepackage{amsmath,amsfonts,amssymb}
\usepackage{xcolor}
\usepackage{makecell}
\newtheorem{prop}{Proposition}

\usepackage{orcidlink}

\begin{document}








\title{Skyline-based exploration of temporal property graphs}

\author{
  Evangelia Tsoukanara\textsuperscript{1,*},
  Georgia Koloniari\textsuperscript{2},
  Evaggelia Pitoura\textsuperscript{3}
}

\institute{Evangelia Tsoukanara \at
              University of Macedonia, Thessaloniki, Greece \\
              \email{etsoukanara@uom.edu.gr}           
           \and
           Georgia Koloniari \at
              University of Macedonia, Thessaloniki, Greece \\
              \email{gkoloniari@uom.edu.gr}
           \and
           Evaggelia Pitoura \at
              University of Ioannina, Ioannina, Greece \\
              \email{pitoura@cse.uoi.gr} 
}

\date{Received: date / Accepted: date}

\maketitle

\begin{flushleft}
  \textsuperscript{1}University of Macedonia, Thessaloniki, Greece\\
  \email{etsoukanara@uom.edu.gr}\\
  ORCID: 0000-0002-0720-7608\orcidlink{0000-0002-0720-7608}\\
  \textsuperscript{2}University of Macedonia, Thessaloniki, Greece\\
  \email{gkoloniari@uom.edu.gr}\\
  ORCID: 0000-0002-9852-7980\orcidlink{0000-0002-9852-7980}\\
  \textsuperscript{3}University of Ioannina, Ioannina, Greece\\ \email{pitoura@cse.uoi.gr}\\
  ORCID: 0000-0002-3775-4995\orcidlink{0000-0002-3775-4995}\\
  \vspace{10pt}
  *Corresponding author: etsoukanara@uom.edu.gr
\end{flushleft}

\newpage


\begin{abstract}

In this paper, we focus on temporal property graphs, that is, property graphs whose labeled nodes and  edges as well as the values of the properties associated with them may change with time. For instance,  consider a bibliographic network, with nodes representing authors and conferences with properties such as gender and location respectively, and edges representing collaboration between authors and publications in conferences. 
A key challenge in studying temporal graphs lies in detecting interesting events in their evolution, defined as time intervals of significant stability, growth, or shrinkage.
To address this challenge, we build aggregated graphs, where nodes are grouped based on the values of their properties, and seek events at the aggregated level, for example, time intervals of significant growth in the collaborations between authors of the same gender. 
To locate such events, we  propose a novel approach based on \textit{unified evolution skylines}. A unified evolution skyline assesses the significance of an event in conjunction with the duration of the interval in which the event occurs.
Significance is measured by a set of counts, where each count refers to the number of graph elements that remain stable, are created, or deleted, for a specific property value.  
For example, for property gender, we measure the number of female-female, female-male, and male-male collaborations. 
Lastly, we share experimental findings that highlight the efficiency and effectiveness of our approach.

\keywords{temporal property graph \and skylines \and exploration \and aggregation \and temporal evolution}
\end{abstract}

\section{Introduction}

Property graphs are often used to model relationships and interactions between real world entities. They consist of  nodes and edges of diverse types, called labels.
Both nodes and edges have properties.  
Property graphs may evolve over time with nodes and edges being added, or deleted. The value of properties may also change with time. We call such graphs \textit{temporal property graphs}.
A key challenge in this context is detecting important events in the evolution of temporal property graphs, such as periods of significant stability, growth or shrinkage. It is interesting to detect such events not only at the individual node and edge level, but also at a higher level. To detect events at a higher level, we build aggregated property graphs where each node corresponds to sets of nodes that have the same property values in the original graph, for a given node label. 

For example, consider a face-to-face proximity graph of students and teachers used for studying the spread of an infectious disease \cite{Gemmetto14}. Nodes are labeled as \textit{individual} with properties such as class, grade, and gender, while edges are labeled as \textit{interact} and correspond to physical interactions between individuals. Detecting events in the evolution of this graph may reveal patterns in the disease spread that will facilitate the application of targeted mitigation strategies. Such patterns may emerge not only at the individual node level, but also at an aggregated one, for example, between students of the same class, or gender. For instance, a pattern may be a correlation between the growth of interactions between 1st graders and an increase in the disease spread, or a correlation between the shrinkage of interactions between boys and a drop in the disease spread.

We adopt an interval-based model, where a temporal property graph is modeled as a property graph whose nodes, edges and properties are annotated with the time intervals during which the corresponding elements were valid. 
We model the evolution of a temporal property graph (including aggregated ones) between two time intervals $T_1$ and $T_2$, where $T_1$ precedes $T_2$ through three \textit{event} graphs: (a) the \textit{stability graph}, that includes nodes and edges that persist in both $T_1$ and $T_2$, (b) the \textit{growth graph} that includes new nodes and edges that did not exist in $T_1$ but appear in $T_2$ and (c)  the \textit{shrinkage graph} that includes deleted nodes and edges that existed in $T_1$ but no longer appear in $T_2$.
Our goal is to detect intervals in each of these graphs when significant events occur, where significance is measured by the number of graph elements that remain stable, are created, or deleted respectively.

Previous work on evolution exploration is driven either by user queries \cite{Rost22}, or
requires a threshold on the number of events \cite{Tsoukanara23}. Both approaches require knowledge of the underlying data to achieve appropriate parameter configuration or create a meaningful query. In this paper, we  adopt a parameter-free novel approach designed for temporal property graphs based on \textit{evolution skylines}.
An evolution skyline is defined on the significance of the event (i.e., number of affected graph elements) and the length of the interval when the event appears. 
An \textit{individual evolution skyline} has two dimensions, namely the length of the interval and the number of affected graph elements with a specific property value, for example, for growth, the number of new interactions i.e., between female students.  \textit{Unified evolution skylines} extend the individual evolution skylines to more than two dimensions. Concretely in a unified evolution skyline, we keep the number of affected graph elements for each possible property value, for example, for growth, we have three numbers, i.e, the number of interactions between female students, between male students and between female and male students.

 
The preferred length of the interval in a skyline depends on semantics.
In this paper, we use two different semantics namely: \textit{strict} and \textit{loose} semantics. Strict semantics capture persistent occurrences of nodes and edges, while loose semantics capture transient changes. 
For example, we are interested in the shortest intervals in the case of loose semantics and shrinkage, so as to capture abrupt changes in the graphs, and in the longest intervals in the case of strict semantics and stability, so as to capture persistence.

In particular, assume that  in the face-to-face
proximity graph,  there is a unified skyline point for property gender with an interval of length
10 and with significance of 30 edges between boys, 20 edges between girls and 15 edges between boys and girls, for the case of stability and
strict semantics. This means that: (a) there is no longer interval where the same number or more edges between boys, or girls, or boys and girls  remained stable, and (b) there is no interval with length 10 or more where more edges between boys, or girls or boys and girls remained stable.

We also extend our unified skyline-based strategy to detect the top-$k$ unified skyline where tuples are ranked based on the number of tuples each one dominates. Finally, we present an experimental evaluation of our approach using three real datasets, aiming at studying both the efficiency and the effectiveness of our approach.


The rest of the paper is structured as follows. In Section 2, we introduce necessary concepts, while in Section 3, we define unified skylines. In Section 4, we present the skyline computation, and Section 5 includes our experimental results. Section 6 summarizes related work and Section 7 offers conclusions.

\section{The temporal property graph model}
A property graph is a directed labeled graph where each node or edge can have a set of properties represented as key-value pairs \cite{Angles18}. \textit{Temporal property graphs} are property graphs that describe the changes occurring in the graph structure and properties over time, capturing the evolution of a graph.



Specifically, we define a temporal property graph based on the TGraph model \cite{Moffitt17}, additionally, supporting property values that can change with time. Let $\mathcal{L}$ be a set of node and edge labels, $\mathcal{P}$  a set of properties, $\mathcal{V}$ a set of property values, and $T$ a set of time intervals, henceforth termed as \textit{temporal element}.

\begin{definition}[Temporal Property Graph] A temporal property graph is defined as a tuple $G[T] = (N, E, \rho, \lambda, \xi^T, \sigma^T)$, where:
\begin{itemize}
\item $N$ is a set of nodes;
\item $E$ is a set of edges;
\item $\rho: E \rightarrow (N \times N)$ is a function that maps an edge to its source and destination nodes;
\item $\lambda: (N \cup E) \rightarrow \mathcal{L}$ is a function that maps a node or a an edge to a label;
\item $\xi^T: (N \cup E) \rightarrow T'$ is a function that maps a node or an edge, to a set $T'$ of time subintervals of $T$, indicating existence of the node or the edge at these subintervals; and
\item $\sigma^T: (N \cup E) \times \mathcal{P} \times t \rightarrow \mathcal{V}$ is a partial function that maps a node or an edge, a property, and a time instance $t \in T$ to a value of the property at this time instance.
 \end{itemize}
\end{definition}

Additionally, it holds that a property can have a value only when the relevant node or edge exists, and that a property has necessarily a value assigned when the corresponding node or edge exists. 
Our model supports both \textit{static} properties for which $\sigma^T$ assigns the same value for all $t \in T$, and \textit{time-varying} properties whose values change during $T$. 

Although our model supports properties for both nodes and edges, in this paper, we focus on node properties. Also, hereafter, we collectively refer to nodes and edges of a graph as \textit{graph elements}. 

Figure \ref{fig:tpg} depicts an example of a temporal property graph for a bibliographic network. We have nodes with two different labels, \textit{author} and \textit{conference}. Authors have two properties, a static one, \textit{gender}, and a time-varying one, \textit{\#publications}, denoting the number of publications of the author per year. Conferences have two properties, one static, \textit{topic}, and one time-varying, \textit{location}. All nodes also include a unique ID property. For simplicity, we associate this property with the temporal elements in which the corresponding node exists. The other property values are also similarly annotated with temporal elements. For example, node with ID 1 exists during [1, 2], is associated with the \textit{author} label, has the static property \textit{gender} with value \textit{male}, and the time-varying property \textit{\#publications} with value 1 for the interval [1], and 2 for the interval [2]. 
Edges in the graph have two labels: \textit{publish} that describe relationships between authors and conferences, and \textit{collaborate} between authors, both also annotated with their corresponding temporal elements. For example, there is a \textit{publish} labeled edge between \textit{author} node with ID 1 and \textit{conference} node with ID 5, denoting that the author published at this conference in interval [2]. Note, that we denote an interval that includes a single time point, e.g. $[1, 1]$, simply as $[1]$.




\begin{figure}
\centering
\includegraphics[scale=0.07]{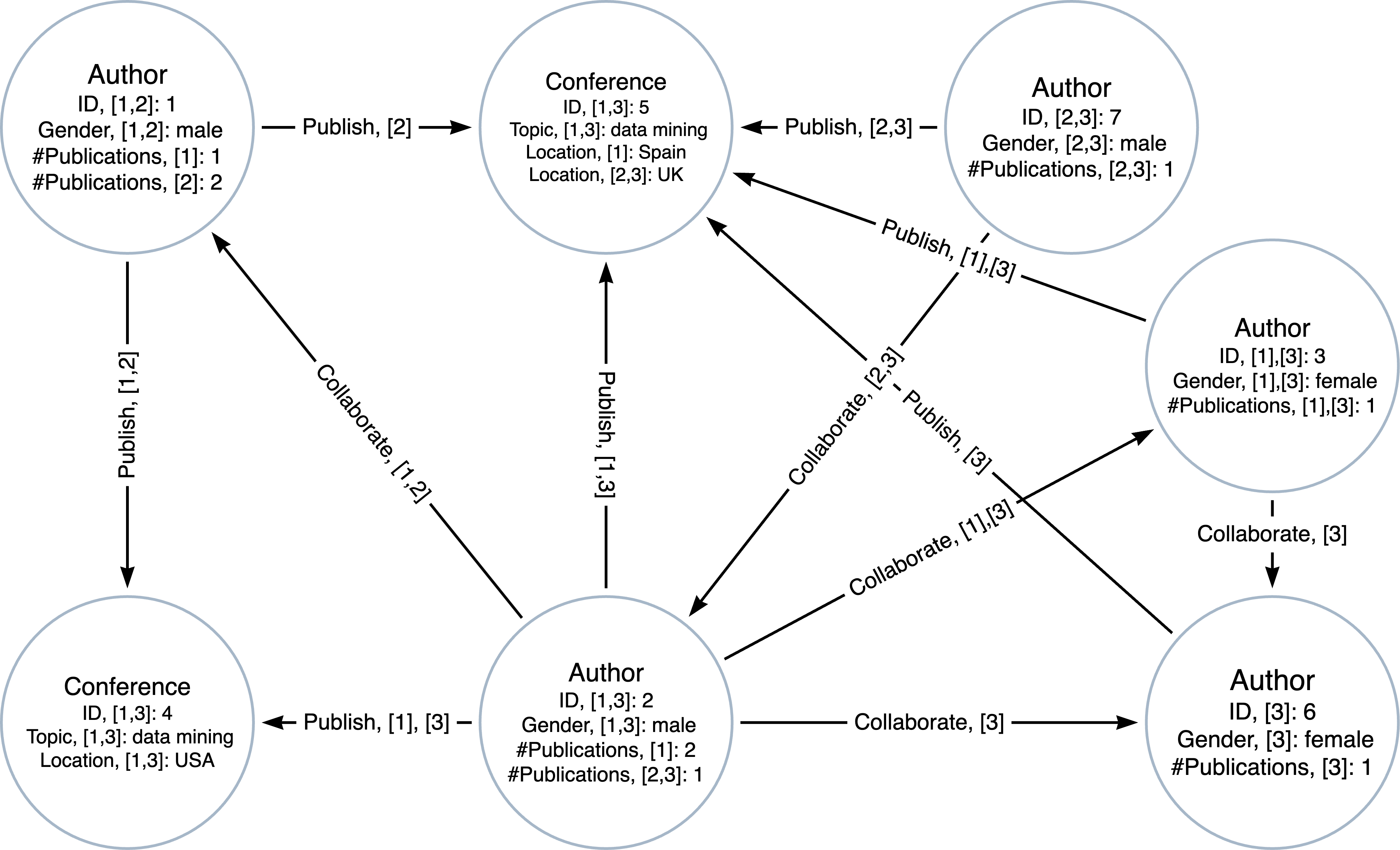}
\caption{An example of a temporal property graph for a bibliographic network defined on $T = [1, 3]$.}
\label{fig:tpg}
\end{figure}

\subsection{Temporal graph operators}
To manage and explore the temporal dimension of our graph model, we define three different set-based operators that take as input two temporal property graphs and produce a new temporal property graph. 

Specifically, given two temporal property graphs $G[T_1]$ and $G[T_2]$, the \textit{union} ($\cup$) operator $G[T_1]\cup G[T_2]$ generates a new temporal property graph with graph elements that appear either in $G[T_1]$ or $G[T_2]$, i.e., graph elements that exist in any time instance of $T_1$ or $T_2$. The \textit{intersection} ($\cap$) operator $G[T_1]\cap G[T_2]$ produces a temporal property graph whose elements appear in both $G[T_1]$ and $G[T_2]$. Finally, the \textit{difference} ($-$) $G[T_1] - G[T_2]$ outputs a temporal property graph that includes graph elements that appear in $G[T_1]$ but not in $G[T_2]$.
Note that the definition of the temporal operators is based on the union, intersection and difference of graph elements, that is, of the sets of nodes and edges. 

\begin{figure}
\centering
\begin{subfigure}[b]{1\textwidth}
\includegraphics[scale=0.07]{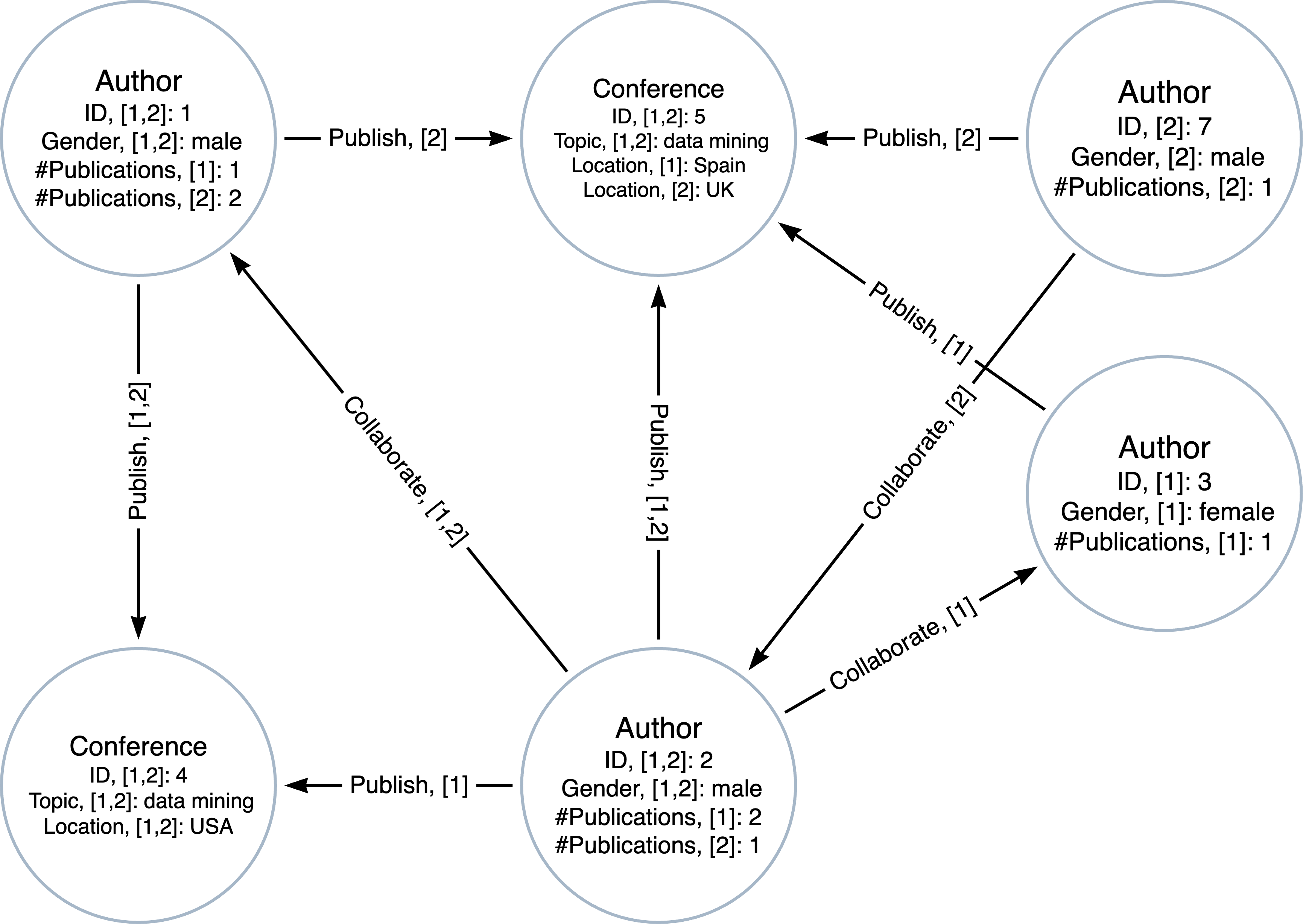}
\centering
\caption{}
\label{fig:sem.1}
\end{subfigure}
\begin{subfigure}[b]{1\textwidth}
\includegraphics[scale=0.07]{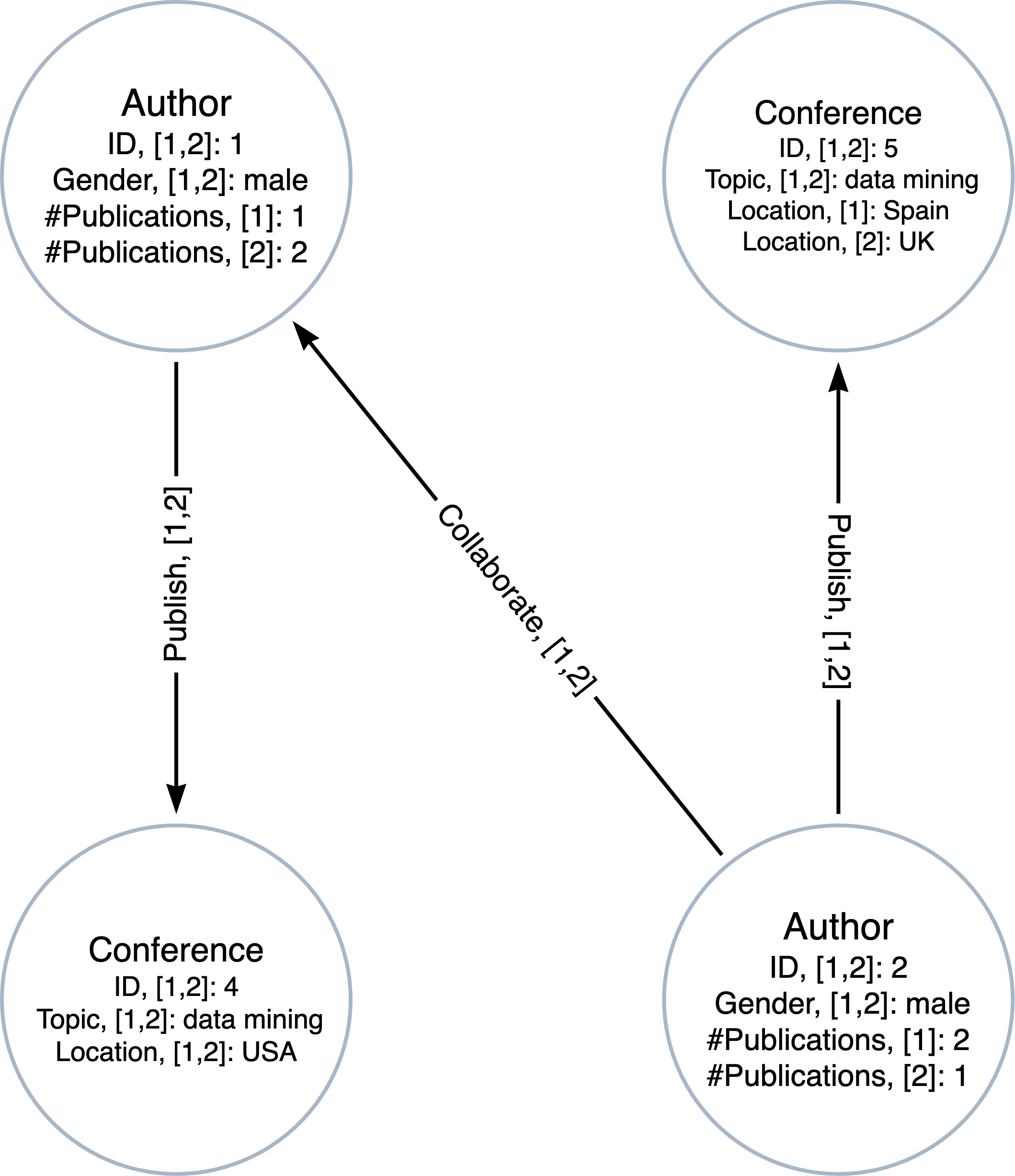}
\centering
\caption{}
\label{fig:sem.2}
\end{subfigure}
\caption{$G[T_1 \sqcup T_2]$ with (a) loose and (b) strict semantics.}
\label{fig:sem}
\end{figure}

In addition, given a  sequence  of graphs  defined at  different temporal elements, we  would like  to combine them  into  a single temporal graph. To this end, we introduce \textit{strict} and \textit{loose} semantics.
Strict semantics retain continuous occurrences of graph elements, while loose semantics retain also transient ones. Intuitively, persistence is associated with the notion of intersection which focuses on finding continuous occurrences, and therefore, consistency over a period. On the other hand, union adopts a more loose strategy and focuses on all graph elements ever appearing during a period of time.

Concretely, for two temporal elements, $T_1$ and $T_2$, let us denote the combined graph as $G[T_1 \sqcup T_2]$. Graph $G[T_1 \sqcup T_2]$ is defined with: (a) \textit{strict semantics} as $G[T_1] \cap G[T_2]$, i.e., we want a graph element to appear in all time instances in $T_1 \cup T_2$  and (b) \textit{loose semantics}, as $G[T_1] \cup G[T_2]$, i.e., we want a graph element to appear in at least one time instance in $T_1 \cup T_2$. 

An example of $G[T_1 \sqcup T_2]$ where $T_1$ is time instance 1 and $T_2$ is time instance 2 defined with loose semantics is shown in Fig. \ref{fig:sem.1} and with strict semantics in Fig. \ref{fig:sem.2}. The graph of Fig. \ref{fig:sem.1} depicts nodes and edges that appear in either time instance 1 or 2, while the graph of Fig. \ref{fig:sem.2} includes only nodes and edges that appear in both 1 and 2. For example, node with ID 3 appears in the graph of Fig. \ref{fig:sem.1} but not in that of Fig. \ref{fig:sem.2} as it exists only in time instance 2.

\subsection{Graph evolution and aggregation}

Our goal is to capture the evolution of a graph between two consequent time intervals $T_1$ and $T_2$ so as to detect whether graph elements remained stable in both intervals, new graph elements appeared in the most recent one, or existing ones disappeared in the most recent one.  Thus, we want to study respectively \textit{stability}, \textit{growth} and \textit{shrinkage}, collectively called \textit{events}.

To model graph evolution with respect to these three types of events, we exploit the temporal operators and define the following three event graphs. For a pair of time intervals $T_1$, $T_2$ where $T_2$ follows $T_1$: (a) the stability graph  ($s$-graph), denoted as  $G_s[(T_2, T_1)]$, is defined as $G[T_1] \cap G[T_2]$ and captures stable graph elements that appear in both $T_1$ and $T_2$,
(b) the shrinkage graph  ($h$-graph),  $G_h[(T_2, T_1)]$, is defined as $G[T_1] - G[T_2]$ and captures elements that disappear from $T_1$ to $T_2$ and (c) the growth graph ($g$-graph), $G_g[(T_2, T_1)]$, is defined as $G[T_2] - G[T_1]$ and captures graph elements that are new to $T_2$ and did not appear in $T_1$. We use $\gamma$ to denote any of the three events.

For example, $G_s[(3, [1 \sqcup 2])]$, for the graph of Fig. \ref{fig:tpg} denotes the part of the graph that did not change between the graph in time interval $[1 \cup 2]$ and the graph in time instance 3 shown in Fig. \ref{fig:t3}.
In case of strict semantics, the graph elements are the nodes and edges that existed in both time instances 1 and 2 (Fig. \ref{fig:sem.2}) and also remain in time instance 3. In case of loose semantics, the graph elements are the nodes and edges that existed in at least one of the time instances 1 and 2 (Fig. \ref{fig:sem.1}) and also remain in time instance 3.


Besides studying the evolution of the graph at the individual node and edge level, we want to study the evolution at a higher granularity level. For example, instead of studying stable interactions between individuals, we may want to study stable interactions between individuals of the same gender, or age. To this end, we employ graph aggregation where given a set of labels $L \subseteq \mathcal{L}$ and properties $P$ associated with these labels, we group nodes that have the same values for these properties, while respecting the network structure.
The resulting property graph consists of two types of nodes: (a) aggregated nodes that correspond to sets of nodes in the original graph whose label belong to $L$ and have the same property values, and (b) nodes of the original graph whose label does not belong to $L$ and thus are not aggregated.

\begin{figure}
\centering
\includegraphics[scale=0.07]{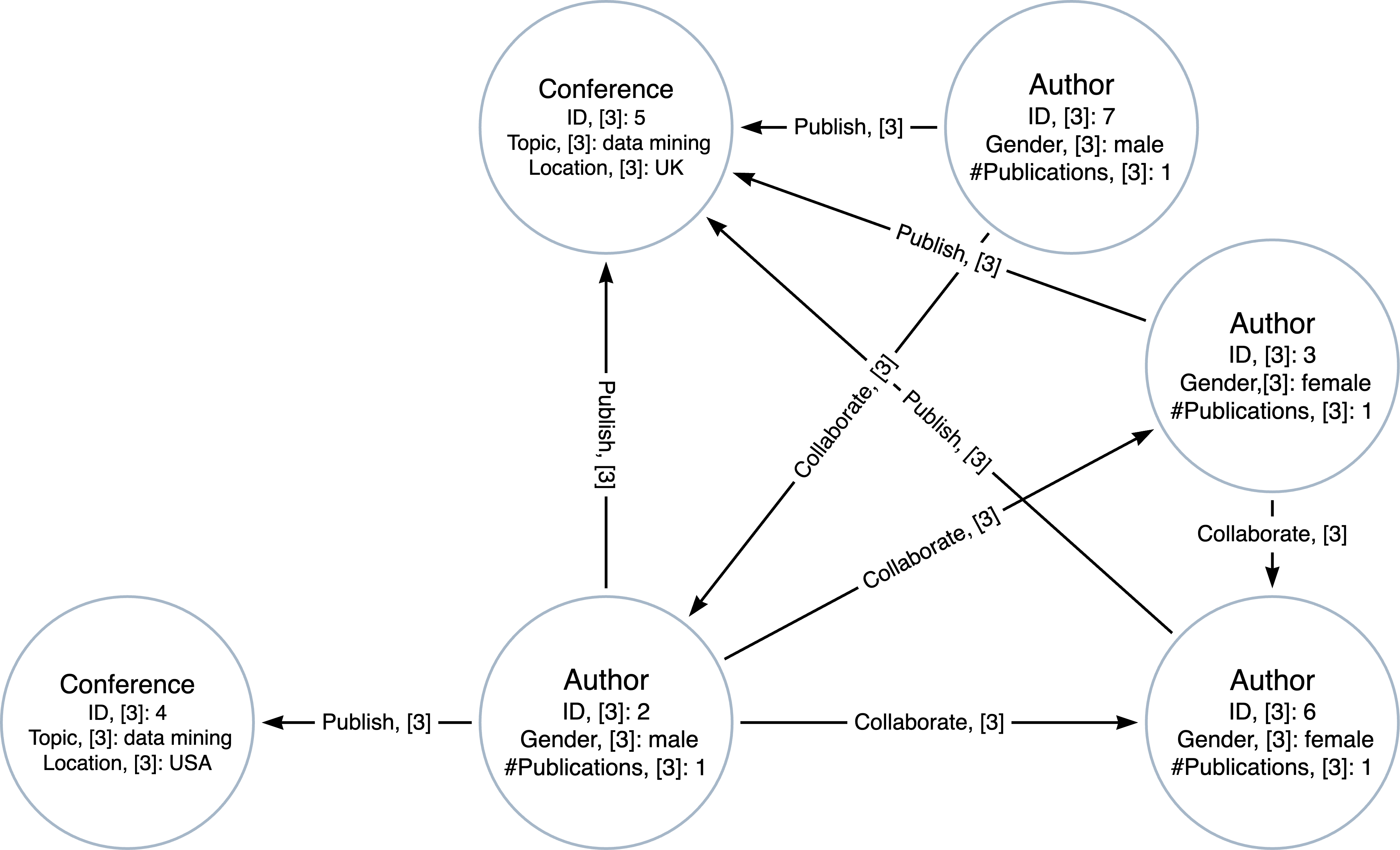}
\caption{Graph instance on time point $3$ for the graph of Fig. \ref{fig:tpg}.}
\label{fig:t3}
\end{figure}

\begin{figure}
\centering
\includegraphics[scale=0.07]{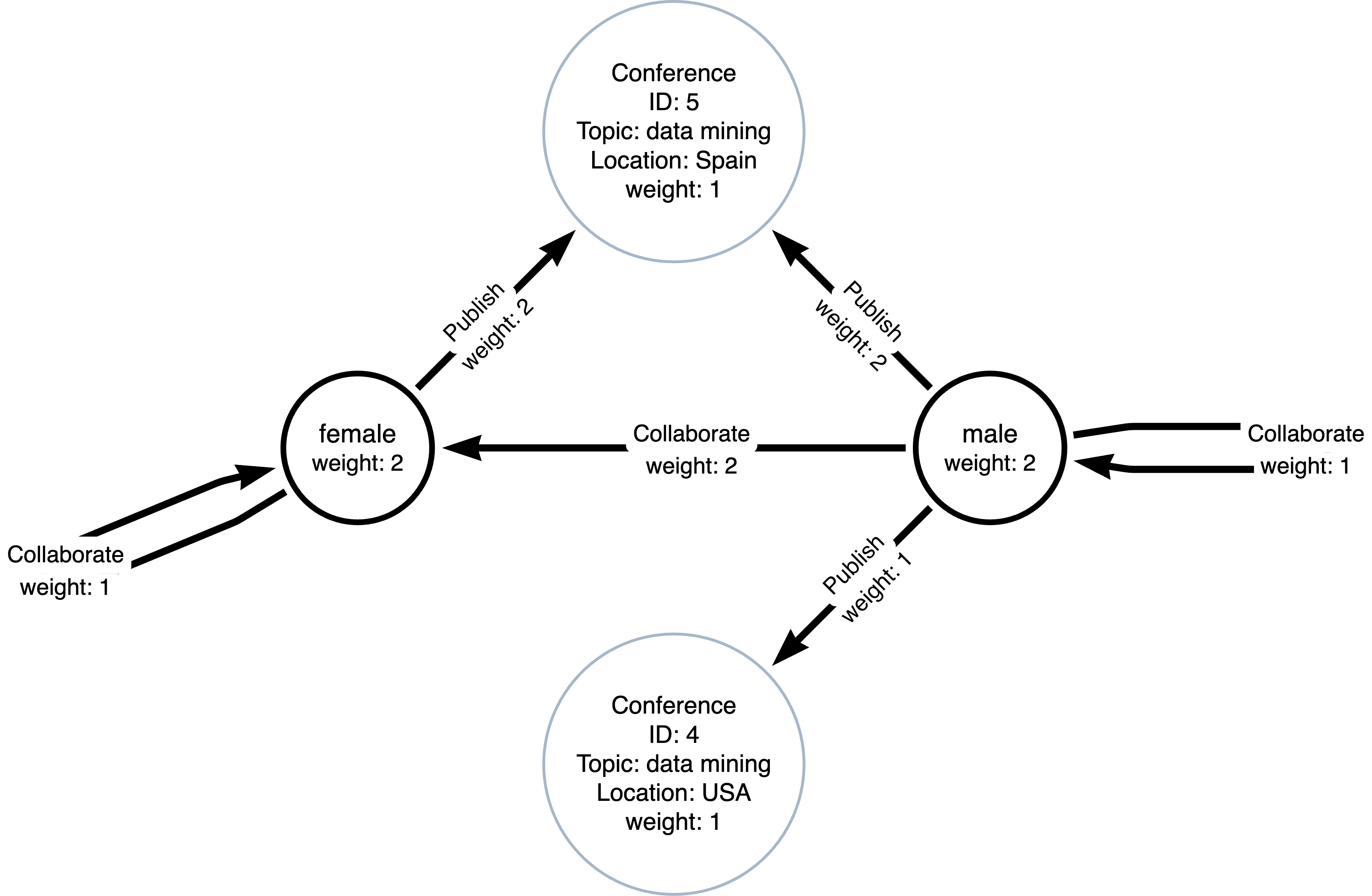}
\caption{Graph $G[[3], \{gender\}]$ aggregated by author label and gender property for the graph of Fig. \ref{fig:t3}.}
\label{fig:agg1}
\end{figure}

\begin{figure}
\centering
\includegraphics[scale=0.07]{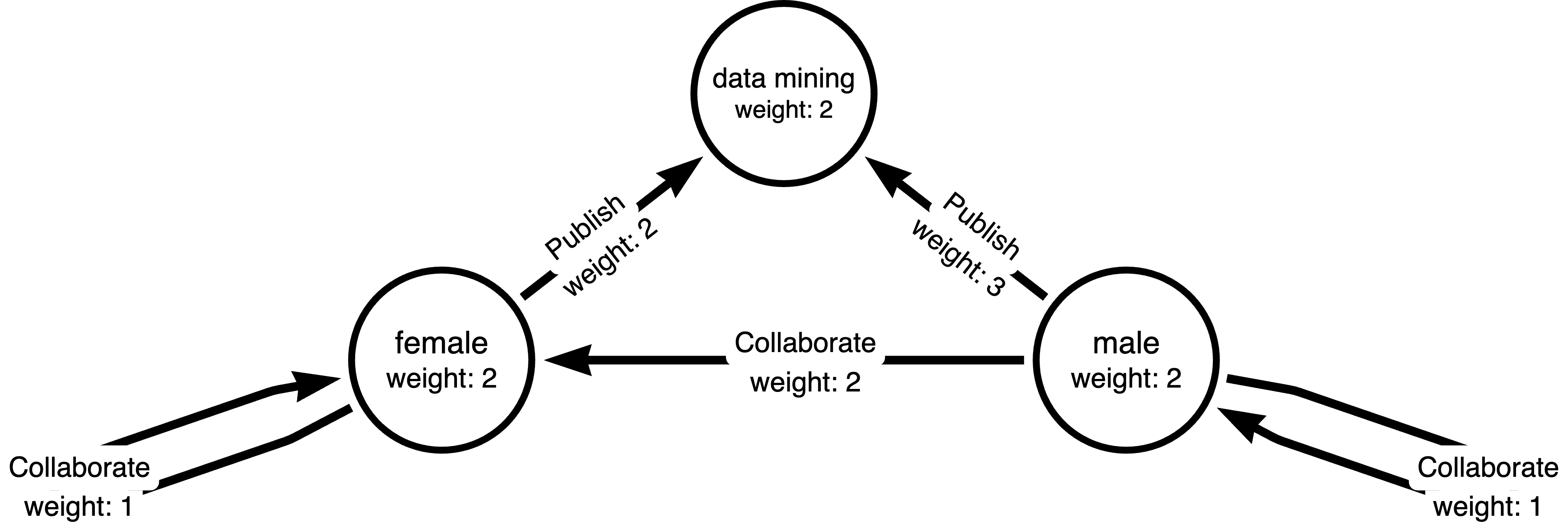}
\caption{Graph $G[[3], \{gender, topic\}]$ aggregated by author, conference labels and gender, topic properties for the graph of Fig. \ref{fig:t3}.}
\label{fig:agg2}
\end{figure}

\begin{definition}[Graph aggregation]
Given a temporal property graph $G[T]$, a set of labels $L \subseteq \mathcal{L}$ and a set of properties $P \subseteq \mathcal{P}$ associated with labels $L$, we define the graph aggregated by $P$ in $T$, $G[T,L:P]$,  as a property graph where:
\begin{itemize}
\item There exists a node $u'$ in $G[T, L:P]$ with label $l_{u'} =p$ for each distinct property value combination $p$ of the $P$ properties, and a single property, \textit{weight} with value equal to the number of distinct nodes in $G[T]$ that have the specific property value combination $p$. Also, there exists a node $v'$  in $G[T, L:P]$ with label $l_v \in \mathcal{L}$, for each node $v$ in $G[T]$ such that $l_v \notin L$, with the same properties and values as node $v$ and an additional property \textit{weight} with value equal to 1.
\item There exists an edge $e'$ with label $l_e \in \mathcal{L}$ between two nodes $u', v'$ in $G[T,L:P]$, if there is an edge of label $l_e$ between the corresponding nodes in the original graph, and a property, \textit{weight}, with value equal to the number of such edges.
\end{itemize}
\end{definition}
For simplicity, we omit labels in the aggregate graph notation using $G[T,P]$ instead of $G[T, L:P]$ to avoid notation overload.




Figure \ref{fig:agg1} depicts the aggregation for the author label, and the gender property for the graph of Fig. \ref{fig:t3}, denoted as $G[[3], \{gender\}]$. Nodes that represent authors are aggregated by gender, whereas the rest of the nodes (conferences) are not affected by the aggregation.  
In more detail, the property \textit{gender} has two values, namely \textit{male}, and \textit{female}, thus there are two aggregated nodes in Fig. \ref{fig:t3}, one with label \textit{male} corresponding to the nodes with ID 2 and 7 in the original graph, and one with label \textit{female} corresponding to nodes with ID 3 and 6. Both nodes have a single property \textit{weight} with value 2. The aggregate graph also includes two nodes that correspond to the nodes with label \textit{conference} in the original graph, i.e., nodes with ID 4 and 5, with an additional property \textit{weight} with value 1.
For nodes of the same property values in the original graph, edges with the same label are grouped together in the aggregated graph and have an additional property \textit{weight} whose value is equal to the number of edges in the original graph that they group.
For instance, there is an edge with label $collaborate$ between nodes labeled $female$ and $male$ corresponding to all collaborations between male and female authors in the original graph, i.e., edges between nodes with ID 2 and 3, and ID 2 and 6, and its property $weight$ has value 2. Similarly, there is an edge with label $publish$ and property \textit{weight} equal to 1, between the node with label $male$ and the $conference$ node with ID 4 corresponding to the edge of the original graph between nodes with ID 2 and 4.


As an additional example, Fig. \ref{fig:agg2} depicts the aggregation for the $author$ and $conference$ labels and their respective properties, $gender$ and $topic$, for the graph of Fig. \ref{fig:t3}, denoted as $G[[3], \{gender, topic\}]$. As all node labels are selected for aggregation, all the nodes of the original graph are aggregated. The aggregated nodes for authors are the same as in Fig. \ref{fig:agg1}. Additionally, conferences are aggregated based on topic. Specifically, there are two conferences of topic data mining in the original graph. Thus, the aggregate graph includes a node for the two conferences with label \textit{data mining} and $weight$ 2. Edges are formed as previously described. For instance, there are 3 publications between male authors and the data mining conferences, so the aggregated graph  contains one corresponding edge with $weight$ equal to 3.

Aggregation can be applied similarly on any event graph. For example, Fig. \ref{fig:events} shows the three event graphs for the evolution of the graph from $G[1 \sqcup 2]$ defined with loose semantics to $G[3]$ aggregated by gender, that is $G_{s}[(3, [1 \sqcup 2]), \{gender\}]$, $G_{g}[(3, [1 \sqcup 2]), \{gender\}]$ and $G_{h}[(3, [1 \sqcup 2]), \{gender\}]$.



To evaluate the significance of an event, we measure the affected graph elements. We focus on events considering edges.
Let $l_e \in \cal{L}$ be an edge label, and $p$ and $p'$ be two property value combinations for a set of properties $P$. 
For an event $\gamma$, we use $count(G_{\gamma}[(T_2, T_1)], l_e, p, p')$ to denote the count of edges with label $l_e$ from nodes with property value combination $p$ to nodes with  property value combination $p'$ in the $\gamma$-graph $G_{\gamma}[(T_2, T_1)]$. 

\begin{prop}
Given an aggregate event graph $G_{\gamma}[(T_2, T_1), P]$, if $p$ and $p'$ correspond to values of properties in $P$, then $count(G_{\gamma}[(T_2, T_1)], l_e, p, p')$ is equal to the weight property of the edge labeled $l_e$ between the nodes with labels $p$ and $p'$ in $G_{\gamma}[(T_2, T_1), P]$. 
\end{prop}
For example, given $G_{s}[(3, [1 \sqcup 2]), \{gender\}]$ in Fig. \ref{fig:events.1}, $count(G_s[3, (1 \sqcup 2))],$ $collaborate, male, female)$, which counts the collaborations between male and female authors that remain stable on $G[3]$ with respect to $G[1 \sqcup 2]$, is equal to the $weight$ property of the \textit{collaborate} edge between the $male$ and $female$ aggregated nodes, and thus, $count(G_s[3, (1 \sqcup 2))], collaborate, male, female)=1$.





\begin{figure*}
\centering
\begin{subfigure}[b]{1\textwidth}
\includegraphics[scale=0.07]{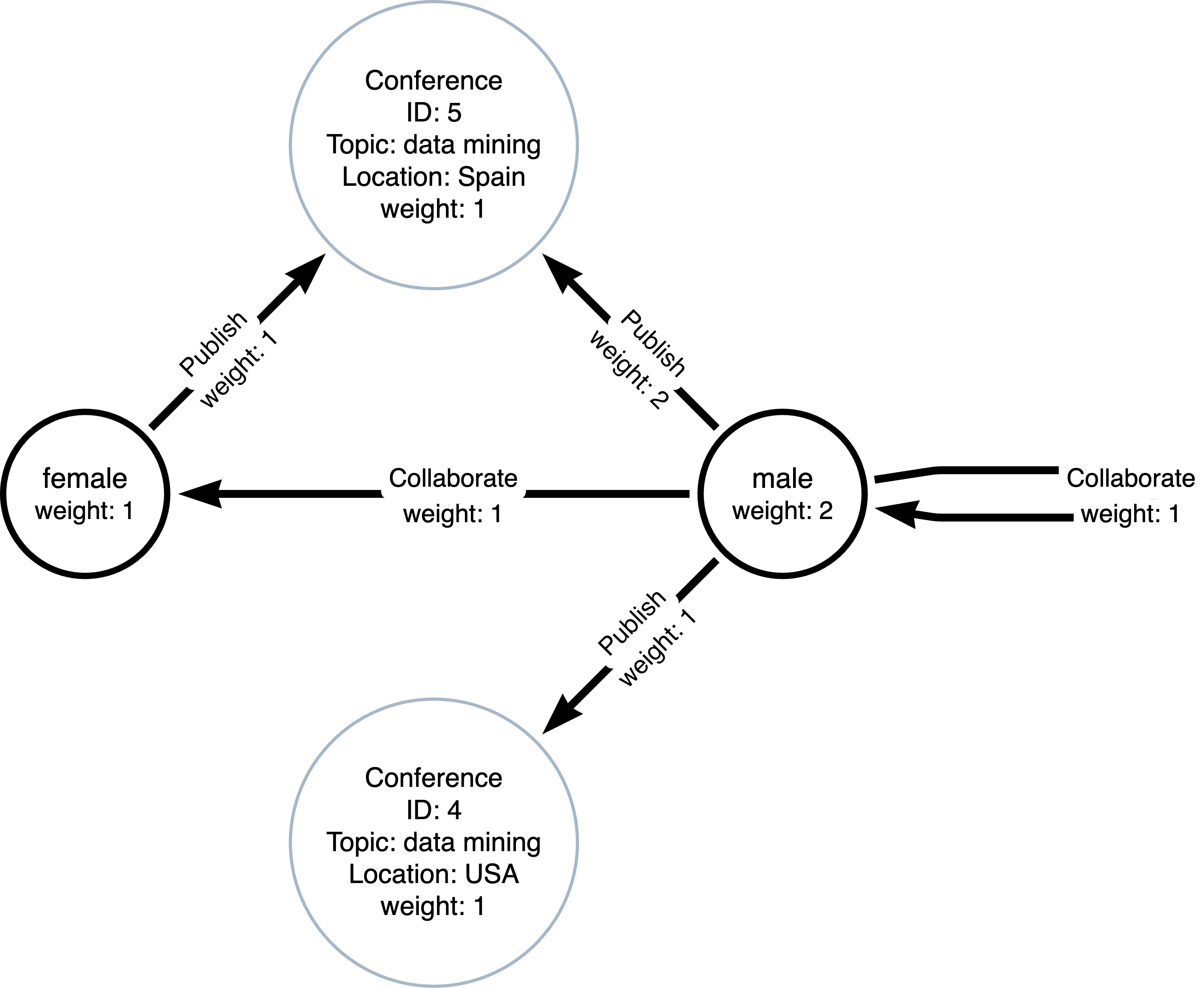}
\centering
\caption{}
\label{fig:events.1}
\end{subfigure}
\begin{subfigure}[b]{0.5\textwidth}
\includegraphics[scale=0.07]{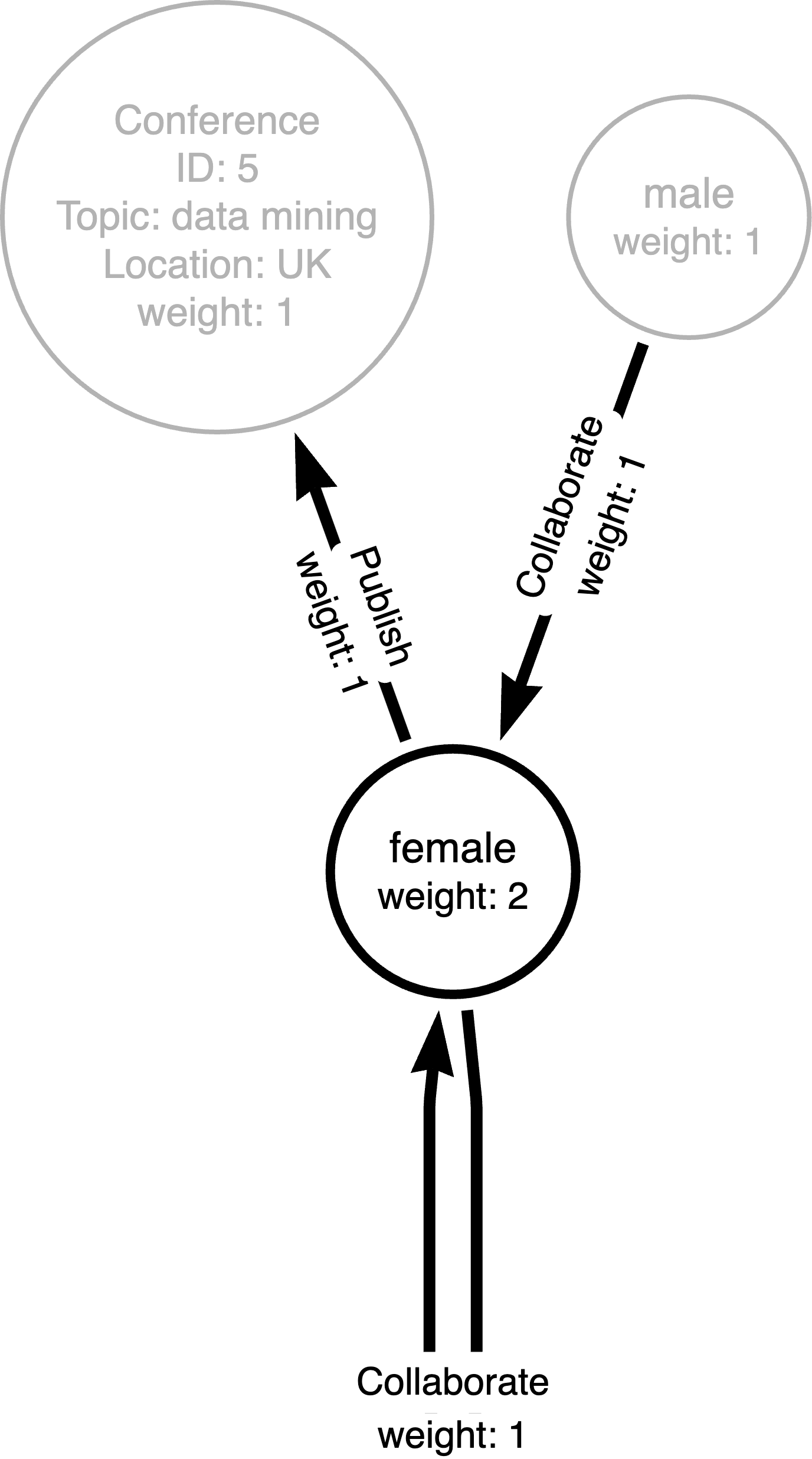}
\centering
\caption{}
\label{fig:events.2}
\end{subfigure}%
\begin{subfigure}[b]{0.5\textwidth}
\includegraphics[scale=0.07]{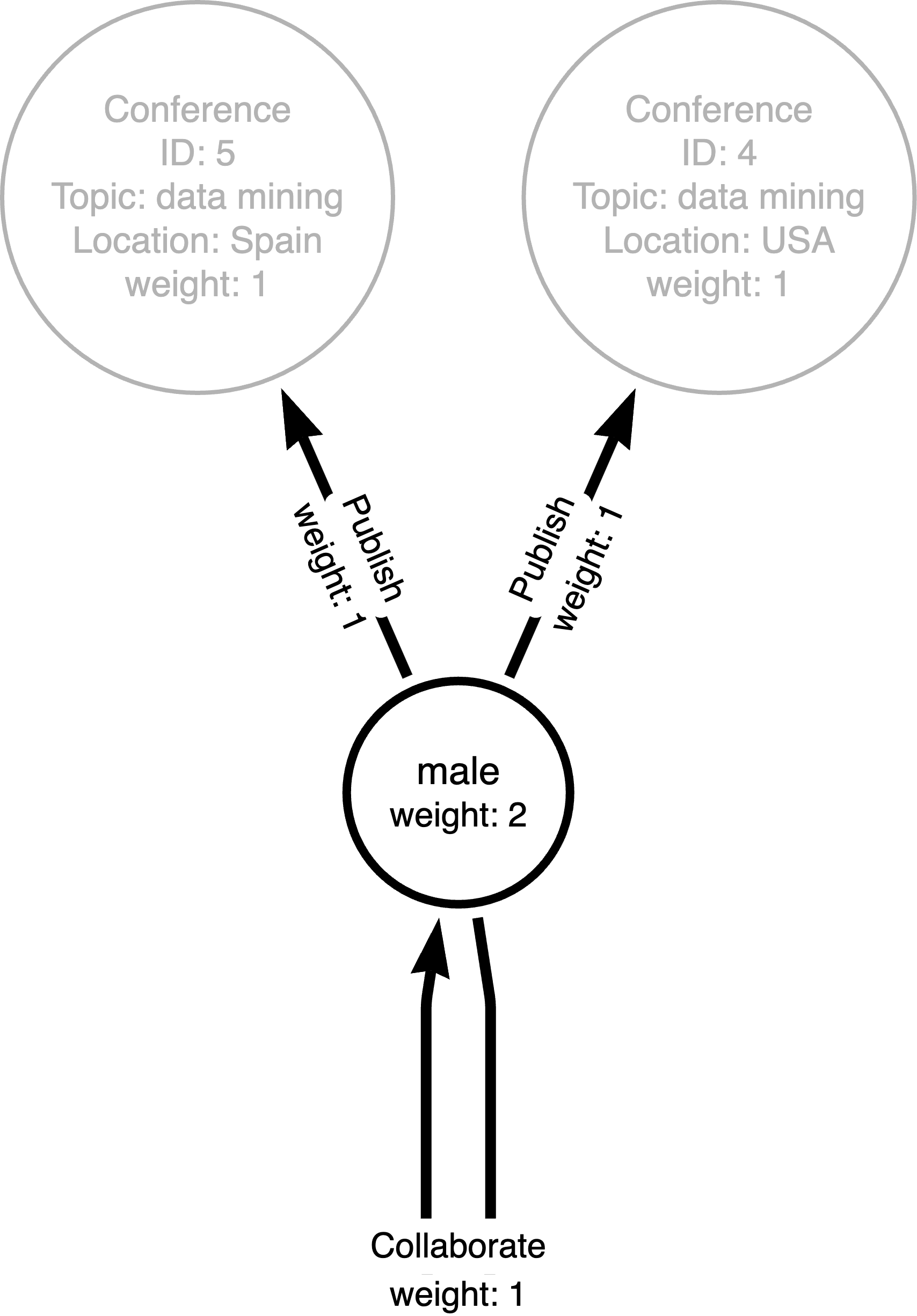}
\centering
\caption{}
\label{fig:events.3}
\end{subfigure}
\caption{(a) Stability ($G_s$), (b) growth ($G_g$) and (c) shrinkage ($G_h$) graphs aggregated on gender for the graph of Fig. \ref{fig:sem.1} to that of Fig. \ref{fig:t3}.}
\label{fig:events}
\end{figure*}

\subsection{Graph exploration}
We would like to identify points in the history of the graph where the count of events is high, indicating interesting points in the evolution of the graph. Specifically, we want to locate time points, termed \textit{reference points}, $t_r$,  where the stability, shrinkage or growth counts, $count(G_{\gamma}[(t_r, T_r)], l_e, p, p')$, are large with respect to an immediate preceding interval $T_r$. To locate such points, for each reference point $t_r$, we perform our search as follows. We initialize $T_r$ to be the time point immediately preceding $t_r$, and then we gradually extend $T_r$ to the past, using either the intersection-inspired \textit{strict} or the union-inspired \textit{loose} semantics.


Given $G_{\gamma}[(t_r, T_r)]$, edge label $l_e$, and property value combinations $p$ and $p'$, we say that $T_r$ is: (i) minimal if $\nexists$ $T'_{r}$ such that $T'_r\subset T_r$ and $count(G_{\gamma}[(t_r, T'_r)],$ $l_e, p, p') \geq count(G_{\gamma}[(t_r, T_r)], l_e, p, p')$, and (ii) maximal if $\nexists$ $T'_{r}$ such that $T_r\subset T'_r$ and  $count(G_{\gamma}[(t_r, T'_r)], l_e, p, p')$ $ \geq count(G_{\gamma}[(t_r, T_r)], l_e, p, p')$.

We say that a count for an event $\gamma$ with respect to strict (loose) semantics, is \textit{increasing} if: for any $T'_r  \supseteq T_r$ defined using strict (resp. loose) semantics, $count(G_{\gamma}[(t_r, T'_r)], l_e, p, p')$ $\geq$ $count(G_{\gamma}[(t_r, T_r)], l_e, p, p')$, for any edge label $l_e$, any combination of property values $p$ and $p'$, and any reference point $t_r$. Similarly, we say that a count for an event $\gamma$ with respect to strict (loose) semantics, is \textit{decreasing} if: for any $T'_r  \subseteq T_r$ defined using strict (resp. loose) semantics, $count(G_{\gamma}[(t_r, T'_r)], l_e, p, p')$ $\leq$ $count(G_{\gamma}[(t_r, T_r)], l_e, p, p')$, for any edge label $l_e$, any combination of property values $p$ and $p'$, and any reference point $t_r$. It holds that count for:
\begin{itemize}
    \item \textit{stability} is \textit{decreasing} wrt strict and \textit{increasing} wrt loose semantics,
    \item \textit{growth} is \textit{increasing} wrt strict and \textit{decreasing} wrt loose semantics, and
    \item \textit{shrinkage} is \textit{decreasing} wrt strict and \textit{increasing} wrt loose semantics.
\end{itemize}

Consequently for increasing counts, we are interested in minimal intervals, while for decreasing counts in maximal intervals.



\section{Evolution skyline}

In the following, we present skyline-based exploration where based on a set of criteria we filter out the exploration results that appear to be of low interestingness.
A skyline approach retrieves all items that have the ``best'' preferred value, in at least one of their features and are not worse in the others. In particular, for multidimensional data, we say that an item dominates another item if it is as good in all dimensions, and better in at least one dimension.  A skyline query returns all non-dominated items in a dataset.

We define two criteria to evaluate the interestingness of an event $\gamma$ in an event graph $G_{\gamma}[(t_r, T_r)]$ with respect to edge label $l_e$ and property value combinations $p$ and $p'$ of properties $P$: the length $len_{T_r}$ of the temporal element $T_r$ defined as the number of time instances $t \in T_r$, and the event count $count(G_{\gamma}[(t_r, T_r)], l_e, p, p')$.

\begin{definition}[Individual evolution skyline] Given an event $\gamma$, a graph $G[T]$, edge label $l_e \in \cal{L}$ and sets of property values, $p$ and $p'$ of some properties $P \subseteq \cal{P}$, the individual evolution skyline is defined as all non-dominated tuples $(t_r, T_r, w)$ where  $t_r$ $\in$ $T$, $T_r$ $\subset$  $T$ and $w$ = $count(G_{\gamma}[(t_r, T_r)],l_e, p, p')$. 
\end{definition}

For an event $\gamma$ with semantics that have increasing counts, we say that a tuple $(t_r, T_r, w)$ dominates a tuple $(t'_{r}, T'_{r}, w')$ if 
$len_{T_r} < len_{T'_{r}}$ and $w \geq w'$, or $w > w'$ and $len_{T_r} \leq len_{T'_{r}}$.
Respectively, for an event $\gamma$ with semantics that have decreasing counts, we say that a tuple $(t_r, T_r, w)$ dominates a tuple $(t'_{r}, T'_{r}, w')$ if 
$len_{T_r} > len_{T'_{r}}$ and $w \geq w'$, or $w > w'$ and $len_{T_r} \geq len_{T'_{r}}$.

For example, considering stability with strict semantics for collaborations between male authors, i.e., edges with $l_e$:\textit{collaborate} and $p,p'$:\textit{male, male}, a tuple $(11, [1\sqcup10], 20)$ in the skyline means that: (a) there is no longer interval where 20 or more such edges remained stable, and (b) there is no interval with length 10 or more where more edges remained stable.

Taking advantage of the increasing and decreasing properties, we restrict the number of such tuples as follows.
\begin{lemma}
Any non-dominated tuple includes an interval $T$ which is either maximal or minimal depending on the use of strict or loose semantics.
\label{lem1}
\end{lemma}

\begin{proof} Without loss of generality, let us consider strict semantics. If the interval $T_r$ of a non-dominated tuple $(t_r, T_r, w)$ is not maximal, then there would exist a pair $(t_r, T'_r, w')$ such that $T_r \subset T'_r$ and for lengths of $T_r, T'_r$, would be $len_{T'_r} > len_{T_r}$ in which $w'\geq w$ events occur. Thus, $(t_r, T'_r, w')$ dominates $(t_r, T_r, w)$.
\end{proof}

When studying evolution, instead of focusing on the events concerning edges of label $l_e$ with particular property values, $p, p'$ as explored by individual skylines, we may be interested in events concerning all value combinations of a particular property set $P\subseteq \mathcal{P}$. 
\begin{definition}[Unified evolution skyline] Given an event $\gamma$, an edge label $l_e \in L$, a set of node properties $P$ and a graph $G[T]$, the unified evolution skyline is defined as all non-dominated tuples $(t_r, T_r, w_1, w_2, \dots, w_n)$, where $t_r \in T$, $T_r$ $\subset$  $T$, and $w_i$ = $count(G_{\gamma}[(t_r, T_r)], l_e, p, p')$ for each of the $n$ property value combinations $p, p'$ of $P$ with label $l_e$.
\end{definition}
Considering all value combinations by exploration with unified skylines allows us to detect points in the evolution of the graph that show interesting events with any property value of some property set $P$. For instance, considering property $gender$, instead of focusing on stable collaborations between, e.g, only male authors (\textit{male, male}), we also consider at the same time ($male, female$), ($female, male$) and ($female, female$) collaborations at the same time. Extending our previous example, for stability with strict semantics and author collaborations ($l_e:collaborate$), a tuple $(11, [1\sqcup10],30,20,15,10)$ in the unified skyline means that: (a) there is no longer interval where the same number or more edges between male authors, or male and female authors, or female and male authors or female authors remained stable, and (b) there is no interval with length 10 or more where more edges between  male authors, or male and female authors, or female and male authors or female authors remained stable. Note that the direction of the edge indicates authorship order.

Similarly to individual skylines, for an event $\gamma$ with semantics that have increasing counts, we say that a tuple $(t_r, T_r, w_1, w_2, \dots, w_n)$ dominates a tuple $(t'_r, T'_r, w'_1, w'_2, \dots, w'_n)$ if it holds that $len_{T_r} < len_{T'_r}$ and $w_i \geq w'_i$ $\forall$ $i$, or $len_{T_r} \leq len_{T'_r}$ and $w_i \geq w'_i$ $\forall$ $i$ and $\exists$ $w_i > w'_i$, where $i \in \{1, \dots, n\}$. Respectively, for an event $\gamma$ with semantics that have decreasing counts, we say that a tuple $(t_r, T_r, w_1, w_2, \dots, w_n)$ dominates a tuple $(t'_r, T'_r, w'_1, w'_2, \dots, w'_n)$ if it holds that $len_{T_r} > len_{T'_r}$ and $w_i \geq w'_i$ $\forall$ $i$, or $len_{T_r} \geq len_{T'_r}$ and $w_i \geq w'_i$ $\forall$ $i$ and $\exists$ $w_i > w'_i$, where $i \in \{1, \dots, n\}$.

From Lemma \ref{lem1} we derive that the intervals included in the non-dominated tuples of the unified evolution skyline are also maximal or minimal.

\begin{lemma}
In the unified skyline, any non-dominated tuple includes an interval $T$ which is either maximal or minimal depending on the use of strict or loose semantics.
\end{lemma}

\begin{proof} Without loss of generality, let us consider a decreasing event. If the interval $T_r$ of a non-dominated tuple with $n$ dimensions $(t_r, T_r, w_1, \dots, w_n)$ is not maximal, then there would exist a tuple $(t_r, T'_r, w'_1, \dots, w'_n)$ such that $T_r \subset T'_r$ and for lengths of $T_r, T'_r$, would be $len_{T'_r} > len_{T_r}$ in which $w'_i\geq w_i, \forall$ $i \in n$. Thus, $(t_r, T'_r, w'_1, \dots, w'_n)$ dominates $(t_r, T_r, w_1, \dots, w_n)$.
\end{proof}

The following lemmas relate individual with unified skylines.

\begin{lemma}
\label{lem3}
A tuple $r=(t_r, T_r, w_1, w_2, \dots, w_n)$ belongs to the unified skyline set $S$, if and only if  $\exists$ $w_i$, $ 1 \leq i \leq n$, for which it holds $(t_r, T_r, w_i) \in s_i$, where $s_i$ is the individual skyline set for the $i$-th combination.
\end{lemma}
\begin{proof} 
Without loss of generality, let us assume the skyline $S$ for an event $\gamma$ with decreasing counts.

Let a tuple $r = (t_r, T_r, w_1, w_2, \dots, w_n)$ $\in S$ such that $\nexists$ $w_i$ for which $(t_r, T_r, w_i)$ $\in s_i$, $ 1 \leq i \leq n$. Since $r\in S$, then $r$ in not dominated by any other tuple. This implies that $\nexists$ $r' = (t_{r'}, T_{r'}, w_1', w_2', \dots, w_n')$ such that either (i) $len_{T_r}<len_{T_{r'}}$ and $w_i \leq w_i', \forall i,  1 \leq i \leq n$, or (ii) $len_{T_r} \leq len_{T_{r'}}$ and $w_i \leq w_i', \forall i,  1 \leq i \leq n$ and $\exists k, 1 \leq k \leq n$, such that $w_k'>w_k$.
If (i) holds, then $(t_r,T_r,w_i)$ is not dominated for all $i,  1 \leq i \leq n$, thus, $(t_r,T_r,w_i) \in s_i$ for all dimensions $i$. Otherwise, if (ii) holds, then  $(t_r,T_r,w_k)$ is not dominated and thus,  $(t_r,T_r,w_k) \in s_k$.

Let $(t_r, T_r, w_k) \in s_k$ and $\nexists r=(t_r, T_r, w_1, \dots, w_k, \dots, w_n) \in S$. Since  $(t_r, T_r, w_k) \in s_k$, it means that $\nexists (t_r'', T_r'', w_k'')$ such that  $len_{T_r''}>len_{T_r}$ and $w_k''\geq w_k$, or $w_k''>w_k$ and $len_{T_r} \leq len{T_r''}$ (1). Also, since  $r \notin S$,  $\exists$ $r' = (t_{r'}, T_{r'}, w_1', w_2', \dots, w_k', \dots, w_n')$ that dominates $r$.  This means that either  (i) $len_{T_r}<len_{T_{r'}}$ and $w_i \leq w_i', \forall i,  1 \leq i \leq n$, or (ii) $len_{T_r} \leq len_{T_{r'}}$ and $w_i \leq w_i', \forall i,  1 \leq i \leq n$ and $\exists y, 1 \leq y \leq n$, such that $w_y'>w_y$. If (i) holds, then $(t_r', T_r', w_k')$ dominates $(t_r, T_r, w_k)$ which is false since $(t_r, T_r, w_k) \in s_k$. Otherwise, if (ii) holds and based on (1) it holds that $len_{T_r}=len{T_r'}$ and $w_k=w_k'$, thus, $r=r'$.

\end{proof}

\begin{lemma}
\label{lem4}
The size of the unified skyline $S$ is at most equal to the sum of the sizes and at least equal to the size of each of the individual skyline sets $(s_1, s_2, \dots, s_n)$ for $n$ property values combination.
\end{lemma}

\begin{proof}
Let us assume that for the size $|S|$ of $S$ it holds that $|S| > (|s_1| + |s_2| + \dots + |s_n|)$. This implies that $\exists$ result $(t_r, T_r, w_1, w_2, \dots, w_n) \in S$ for which $(t_r, T_r, w_i) \notin s_i$, $\forall$ $w_i \in$ $(w_1, w_2, \dots, w_n)$, which is false according to Lemma \ref{lem3} as $\forall$ $r \in S$ $\exists$ $w_i$ for which holds $(t_r, T_r, w_i) \in s_i$. Also let us assume that $\exists$ $s_i$ $\in (s_1, s_2, \dots, s_n)$ for which $|S| < |s_i|$. This implies $\exists$ $(t_r, T_r, w_i) \in s_i$ which is not included in any of the results $r$ $\in S$. The assumption is false according to Lemma \ref{lem3} as every result in the skyline set is necessarily included in the results of the unified skyline set.
\end{proof}

As skylines may contain many results, especially when considering graphs evolving over long periods of time and multiple property value combinations, we extend our definition of the unified evolution skyline to the \textit{top}-$k$ \textit{unified evolution skyline}. To this end, we define the \textit{domination degree}.
\begin{definition}[Domination Degree]
Given an event $\gamma$ and edge label $l_e$, the domination degree of a tuple $r=(t_r, T_r, w_1, \dots, w_n)$, $dod(r)$, is defined as the number of tuples $r$ dominates. 
\end{definition}
To determine the top-$k$ skylines, we rank all $(t_r, T_r, w_1, \dots, w_n)$ tuples in the unified evolution skyline according to their domination degree and retrieve only the top-$k$ ranking ones.

\section{Skyline-based exploration}

We now describe the evolution skyline computation process, that given an event $\gamma$, a temporal property graph $G[T]$, an edge label $l_e$ and a set of properties $P \subseteq \mathcal{P}$, computes all non-dominated tuples $(t_r, T_r, W)$, $t_r \in T, T_r \subset T$ and $W$ is a hash table containing the counts $count(G_{\gamma}[(t_r, T_r)],$ $l_e, p, p')$ for each $p, p'$ of $P$.


A naive approach to compute the skyline would entail two phases. First, computing all possible $(t_r, T_r, W)$, i.e., for each $t_r \in T$ enumerate all $T_r \subset T$ and computing the corresponding $W$, and then comparing all these tuples so as to eliminate all the dominated ones. However, this exhaustive approach can be easily improved by pruning all dominated results as they are derived. We propose incrementally evaluating the skyline for each reference point $t_r$, while at the same time updating the final skyline as new results are derived.

\begin{algorithm}
	\KwIn{A temporal property graph $G[T]$, an edge label $l_e$, a set of properties $P$}
	\KwOut{Skyline set $S_{sky}$, dominance set $S_{dom}$}
    Initialize hash tables $S_{sky}: len_{T_r} \rightarrow{\{(t_r, T_r, W)\}}$, $S_{dom}: (t_r, T_r, W) \rightarrow{|d|}$, $d:$ set of dominated items\\
    
    \For{each reference point $t_r$}{
        $T_r \leftarrow $ longest possible interval for $t_r$\\
        $len_{T_r}, len_{T_{r_{max}}} \leftarrow length(T_r)$\\
        \While{$len_{T_r} \geq 1$}{
            Initialize hash table $W_{curr}$\\
            Compute $Intersection$ $G_s[(t_r, T_r)] = G[t_r] \cap G[T_r]$\\
            Compute $Aggregation$ $G_s[(t_r, T_r), P]$\\
            \For{each property value combination $p, p'$ of $P$}{
                $W_{curr}[(p,p')] = count(G_s[(t_r, T_r)], l_e, p, p')$
            }
            $W_{curr}[length] = len_{T_r}$\\
            $curr \leftarrow (t_r, T_r, W_{curr})$\\
            $S_{dom}[curr] \leftarrow 0$\\
            $len_{prev} \leftarrow len_{T_r}$\\            
            \While{$\nexists$ $S_{sky}[len_{prev}]$ and $len_{prev} \leq len_{T_{r_{max}}}$}{
                $len_{prev}$ $+= 1$
            }
            \If{$len_{prev} > len_{T_{r_{max}}}$}{
                $S_{sky}[len_{T_r}] \leftarrow curr$\\
                $S_{dom}[curr]$ $+= 1$\\
            }
            \Else{
                \For{each tuple $prev$ $\in$ $S_{sky}[len_{prev}]$}{
                    \If{for all $(p, p')$ $W_{curr}[(p, p')] >= W_{prev}[(p, p')]$ $and$ $len_{T_r} = len_{prev}$ $and$ $\exists$ $W_{curr}[(p, p')] > W_{prev}[(p, p')]$}{
                        $S_{dom}[curr]$ $+= 1$\\
                        $S_{dom}[curr]$ $+= S_{dom}[prev]$\\
                        $remove$ $prev$ $from$ $S_{sky}[len_{prev}]$\\
                        $S_{sky}[len_{prev}] \leftarrow curr$
                    }
                    \ElseIf{for all counts $W_{curr}[(p, p')] <= W_{prev}[(p, p')]$ $and$ $len_{T_r} \leq len_{prev}$ $and$ $(\exists$ $W_{curr}[(p, p')] < W_{prev}[(p, p')]$ $or$ $len_{T_r} < len_{prev})$}{
                        $S_{dom}[prev]$ $+= 1$
                    }
                    \Else{
                        $S_{sky}[len_{prev}].append(curr)$
                    }
                    \If{$len_{T_r} > 1$}{
                        $len_{prev} \leftarrow len_{T_r} - 1$\\
                        \While{$\nexists$ $S_{sky}[len_{prev}]$ and $len_{prev} \geq 1$}{
                            $len_{prev}$ $-= 1$
                        }        
                        \If{for all counts $W_{curr}[(p, p')] >= W_{prev}[(p, p')]$}{
                            $S_{dom}[curr]$ $+= 1$\\
                            $S_{dom}[curr]$ $+= S_{dom}[prev]$\\
                            $remove$ $prev$ $from$ $S_{sky}[len_{prev}]$
                        }
                    }
                }
            }
        Reduce $T_r$\\
        $len_{T_r}$ $-= 1$\\
        }
    }
	\Return{$S_{sky}, S_{dom}$}
	\caption{{\sc Skyline-based Stability ($\cap$) Exploration}}
	\label{algo:SkyExStabMax}
\end{algorithm}

Thus, given a reference point $t_r$, we enumerate all possible $T_r$ to compute our candidate tuples. The order in which we enumerate the temporal elements depends on the semantics used. In particular, when the event $\gamma$ has increasing counts and the skyline consists of minimal temporal elements, then we start with the shortest temporal element and gradually extend it using either strict or loose semantics. Similarly, when $\gamma$ has decreasing counts and the skyline consists of maximal temporal elements, then we start with the longest temporal element and gradually extend it using either strict or loose semantics. 

The process repeats iteratively for each $t_r, T_r$ pair. First, the event graph $G_{\gamma}[(t_r, T_r)]$ is computed, and then 
is aggregated on $P$ producing $G_{\gamma}[(t_r, T_r), P]$. 
To compute the event graph and its aggregation, we leverage the methods introduced in the GraphTempo framework \cite{Tsoukanara23} that provides efficient algorithms for the implementation of all temporal operators and property aggregation differentiating between static and time-varying properties to handle both more efficiently.

Every derived tuple $(t_r, T_r, W)$ is considered as a candidate for the skyline and compared against previous candidate tuples. If it is dominated, it is ignored and we proceed with the next shorter or longer $T'_r$. Otherwise, $(t_r, T_r, W)$ is added to the skyline and any previously added tuple $(t'_{r}, T'_{r}, W')$ that it dominates is pruned from the skyline.

The process for skyline evaluation can be extended to support $top$-$k$ skyline evaluation if we evaluate and maintain the domination degree of each  $(t_r, T_r, W)$ that is added in our skylines. Initially, $dod(t_r, T_r, W)$ is set to 0. When $(t_r, T_r, W)$ is added to the skyline, if it substitutes other tuples already in the skyline, its $dod(t_r, T_r, W)$ is set to the sum of the $dod$s of all tuples it substitutes.

Algorithm \ref{algo:SkyExStabMax} presents in detail the implementation of the evolution skyline for the event of stability ($\cap$) where we are interested in maximal temporal elements. We next prove that Algorithm \ref{algo:SkyExStabMax} correctly computes the skyline, while similar proofs can be derived for other events and semantics.
\begin{lemma}
If a tuple belongs to the skyline,  Algorithm  \ref{algo:SkyExStabMax} will return it, and vice versa, any result returned by  Algorithm  \ref{algo:SkyExStabMax} belongs to the skyline.
\end{lemma}
\begin{proof}
Algorithm  \ref{algo:SkyExStabMax} considers  skylines with  maximal temporal elements. 
Assume for the purposes of contradiction, that there is a tuple $s = (t_r, T_r, W)$ that belongs to the skyline but is not returned by the algorithm.
Since, the algorithm considers for all reference points, all intervals, this means that $s$ was considered and it was either (a) not inserted in the skyline in lines 27-28
or (b) was pruned in lines 22-26 or in lines 35-38.
If it was not inserted in lines 27-28, it means
that is was dominated by a previous tuple with all counts and length at least equal, and either one of the counts larger or length larger, thus $s$ does not belong to the skyline, which is a contradiction. If it was pruned in lines 22-26, it means that the new tuple has all counts at least equal, one of the counts is larger and has the same length, thus $s$ does not belong to the skyline.
If it was pruned in lines 35-38, it means that the new tuple has all counts at least equal and longer interval, thus $s$ does not belong to the skyline.
Similarly, suppose a tuple $s'=(t'_r, T'_r, W') \notin S_{sky}$ returned by Algorithm \ref{algo:SkyExStabMax}. Then, it holds that $\exists$ $s = (t_r, T_r, W)$ that dominates $s'$. If $s$ is examined after $s'$, if $len_{T_r} = len_{T'_r}$, the check of lines 22-26 would prune $s'$ as $s$ would have all counts at least equal to those of $s'$ and at least one of them larger, or if $len_{T_r} > len_{T'_r}$ the check of lines 35-38 would prune $s'$ with counts of $s$ being at least equal to those of $s'$. Otherwise, if $s'$ is examined after $s$, lines 27-28 would prune $s'$. Thus, $s' \in S_{sky}$.
\end{proof}

\section{Experimental evaluation}

For our experiments, we use three real-world datasets.

\textit{\textbf{DBLP}}: a bibliographic network extracted from DBLP\footnote{https://dblp.uni-trier.de/}. The dataset contains 2 labels for nodes, \textit{author}, \textit{conference} and 2 labels for edges, \textit{collaborate}, \textit{publish}. The dataset covers a period of 21 years from 2000 to 2020, where each year denotes a time point. There is an edge of label collaborate when two authors co-author a paper in a year where the direction of the edge indicates the order of the authors of the paper. Also, there is an edge of label publish when an author publishes to a conference at a specific year. There are 2 properties for the authors, one static, \textit{gender} with 2 values, and one time-varying, \textit{\#publications} per year with 3 categorical values (low, average, high), and 2 properties for the conferences, one static, \textit{topic} with 16 values, and one time-varying, \textit{location} with 50 values. Gender is determined utilizing an appropriate tool\footnote{https://github.com/lead-ratings/gender-guesser}, topic is inferred from the name of the conference and location is derived from DBLP.

\textit{\textbf{MovieLens}}: a movie ratings datasets build on the benchmark MovieLens \cite{Harper16}. There are nodes and edges of a single label, \textit{user}, and \textit{co-rate}, respectively. The datasets covers a period of 6 months from May 1st to October 31st of 2000, and each month represents a time point. There is an edge of label co-rate between two users when they have both rated the same movie. The direction on edges denotes the precedence of a rating. Each user has 3 static properties, \textit{gender} with 2 values, \textit{age} with 7 categorical values, \textit{occupation} with 21 values, and 1 time-varying, \textit{average rating} of the user per month with 38 values.

\textit{\textbf{Primary School}}: a contact network describing the face-to-face proximity of 232 students and 10 teachers of a primary school in Lyon, France. The dataset has nodes and edges of a single label, that is \textit{individual} and \textit{interact}, respectively. Each interaction denotes a 20-second contact between two individuals. Each individual has 2 static properties, that is \textit{gender}, \textit{class}. The school has 5 grades, 1 to 5, with 2 classes each (i.e., 1A, 1B, 2A, 2B, etc) plus teachers, and 3 values for gender property, female, male, and unspecified. Our dataset covers a period of 17 hours.

Tables \ref{tab:dblp}, \ref{tab:ml}, and \ref{tab:pm} show the size of the graphs per time point for the three datasets.

Our methods are implemented in Python 3.9 and our experiments are conducted in a Windows 10 machine with Intel Core i5-2430, 2.40GHz processor and 8GB RAM. Our code and data are publicly available\footnote{https://github.com/etsoukanara/uskylinexplore}.

\begin{table*}
\caption{\textit{DBLP} property graph}
\label{tab:dblp}
\centering
\begin{adjustbox}{width=1\textwidth}
\begin{tabular}{c|c|ccccccccccccccccccccc}
\hline
\centering
{} & \textbf{Type} & 2000 & 2001 & 2002 & 2003 & 2004 & 2005 & 2006 & 2007 & 2008 & 2009 & 2010 & 2011 & 2012 & 2013 & 2014 & 2015 & 2016 & 2017 & 2018 & 2019 & 2020\\ \hline
\multirow{2}{*}{Nodes} & {author} & 1708 & 2165 & 1761 & 2827 & 3278 & 4466 & 4730 & 5193 & 5501 & 5363 & 6236 & 6535 & 6769 & 7457 & 7035 & 8581 & 8966 & 9660 & 11037 & 12377 & 12996\\
{} & {conference} & 15 & 18 & 20 & 19 & 24 & 25 & 25 & 25 & 23 & 27 & 27 & 29 & 25 & 29 & 27 & 28 & 28 & 32 & 35 & 34 & 28\\ \hline
\multirow{2}{*}{Edges} & {collaborate} & 2336 & 2949 & 2458 & 4130 & 4821 & 7145 & 7296 & 7620 & 8528 & 8740 & 10163 & 10090 & 11871 & 12989 & 12072 & 15844 & 16873 & 18470 & 21197 & 27455 & 28546\\
{} & {publish} & 1902 & 2402 & 2067 & 3188 & 3789 & 5266 & 5553 & 6140 & 6419 & 6445 & 7278 & 7874 & 7954 & 8891 & 8338 & 10205 & 10617 & 11538 & 13463 & 15203 & 16246\\
\hline
\end{tabular}
\end{adjustbox}
\end{table*}

\begin{table}
\caption{\textit{MovieLens} property graph}
\label{tab:ml}
\centering
\resizebox{6.5cm}{!}{%
\begin{tabular}{c|c|cccccc}
\hline
\centering
{} & \textbf{Type} & May & Jun & Jul & Aug & Sep & Oct\\ \hline
{Nodes} & user & 486 & 508 & 778 & 1309 & 575 & 498\\ \hline
{Edges} & co-rate & 100202 & 85334 & 201800 & 610050 & 77216 & 48516\\
\hline
\end{tabular}
}
\end{table}

\begin{table*}
\caption{\textit{Primary School} property graph}
\label{tab:pm}
\centering
\begin{adjustbox}{width=1\textwidth}
\begin{tabular}{c|c|ccccccccccccccccc}
\hline
\centering
{} & \textbf{Type} & 1 & 2 & 3 & 4 & 5 & 6 & 7 & 8 & 9 & 10 & 11 & 12 & 13 & 14 & 15 & 16 & 17\\ \hline
{Nodes} & individual & 228 & 231 & 233 & 220 & 118 & 217 & 215 & 232 & 238 & 235 & 235 & 236 & 147 & 119 & 211 & 175 & 187\\ \hline
{Edges} & interact & 857 & 2124 & 1765 & 1890 & 1253 & 1560 & 1051 & 1971 & 1170 & 1230 & 2039 & 1556 & 1654 & 1336 & 1457 & 1065 & 1767\\
\hline
\end{tabular}
\end{adjustbox}
\end{table*}

\subsection{Performance evaluation}

In the first set of experiments we study the execution time for the unified skyline exploration for the events of stability ($G_s(\cap)$, $G_s(\cup)$), growth ($G_g(\cup)$), and shrinkage ($G_h(\cup)$).
We use $\cup$ to denote loose semantics, and $\cap$ strict. We aggregate on various node labels and study the skylines for various edge labels.

In Figures \ref{fig:dblptime}, \ref{fig:mltime} and  \ref{fig:pstime}  we report the runtime for computing the skylines as
we increase the size of the graph, considering intervals up to a specific reference point for the \textit{DBLP}, \textit{MovieLens} and \textit{Primary School}.
We notice that events with loose semantics need longer time to execute compared to strict ones. Also, we observe that the most time-consuming event overall is shrinkage.

In Fig. \ref{fig:dblptime.1}, we study collaborations for aggregation by gender, while in Fig. \ref{fig:dblptime.3}  publications for aggregation by (gender, topic) for the  \textit{DBLP} dataset. When aggregating by gender, the resulting graph includes aggregate nodes for authors grouped by gender, while the rest of the nodes in the aggregated graph, i.e. conferences, remain the same as in the original graph. When aggregating by both gender for the authors and topic for the conferences, the entire graph consists of grouped nodes. Consequently, the aggregated graph produced in Fig. \ref{fig:dblptime.1} is significantly larger comparing to that in Fig. \ref{fig:dblptime.3}, resulting in higher execution time for all kinds of events, especially impacting $G_s(\cup)$.

Moving on to \textit{MovieLens}, Fig. \ref{fig:mltime} depicts the execution time for the exploration of co-rating when aggregating by gender (Fig. \ref{fig:mltime.1}), age (Fig. \ref{fig:mltime.2}), and (gender, age) (Fig. \ref{fig:mltime.3}). Despite \textit{MovieLens} having a significantly larger number of edges compared to \textit{DBLP}, the exploration of all events requires less time. This suggests that the number of time points is a crucial factor influencing the execution time of our algorithm. 

Finally, for \textit{Primary School}, Fig. \ref{fig:pstime} depicts the execution time for the interactions when aggregating by gender (Fig. \ref{fig:pstime.1}), class (Fig. \ref{fig:pstime.2}), and (gender, class) (Fig. \ref{fig:pstime.3}). Notably, the \textit{Primary School} graph proves to be faster to explore for any type of event. This is attributed to its smaller size with a limited number of nodes and edges.

In conclusion, our findings highlight that exploration time is closely linked to both the number of time points and the size of the graph.


\begin{figure*}
\centering
\begin{subfigure}[b]{0.5\textwidth}
\includegraphics[scale=0.21]{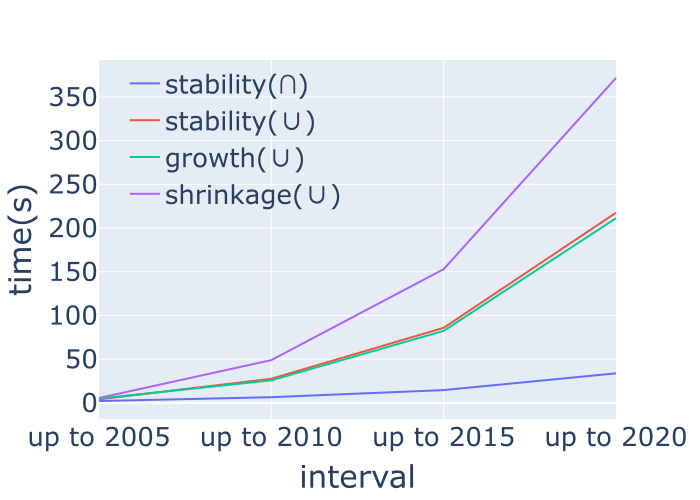}
\centering
\caption{}
\label{fig:dblptime.1}
\end{subfigure}%
\begin{subfigure}[b]{0.5\textwidth}
\includegraphics[scale=0.21]{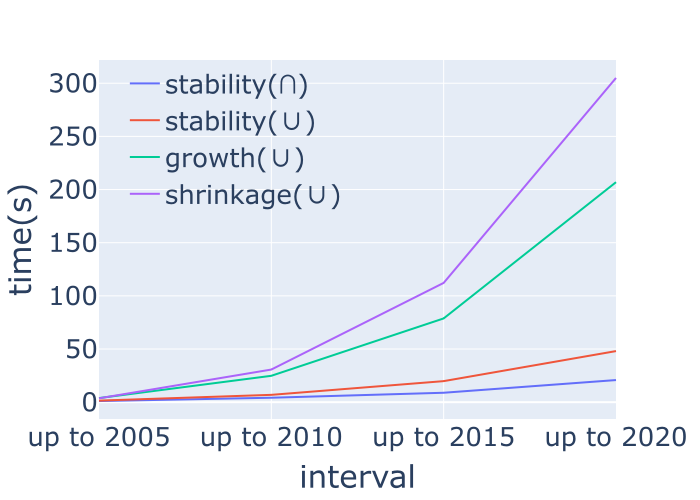}
\centering
\caption{}
\label{fig:dblptime.3}%
\end{subfigure}
\caption{Skyline exploration execution time for \textit{DBLP} for labels (a) collaborate and aggregation by gender, (b) publish and aggregation by (gender, topic).}
\label{fig:dblptime}
\end{figure*}

\begin{figure*}
\centering
\begin{subfigure}[b]{0.33\textwidth}
\includegraphics[scale=0.21]{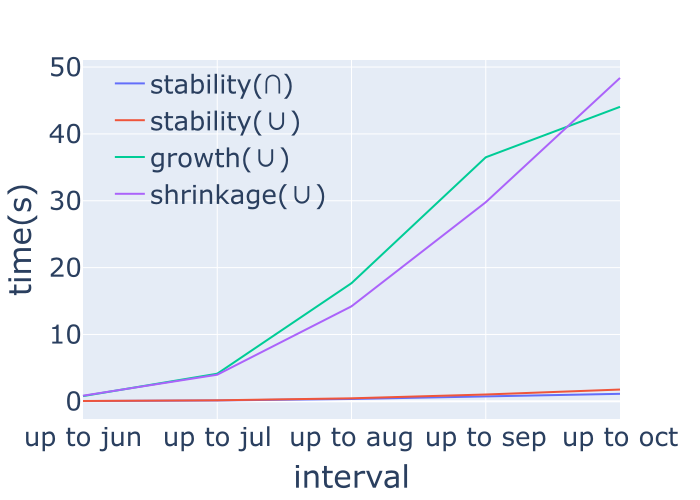}
\centering
\caption{}
\label{fig:mltime.1}
\end{subfigure}%
\begin{subfigure}[b]{0.33\textwidth}
\includegraphics[scale=0.21]{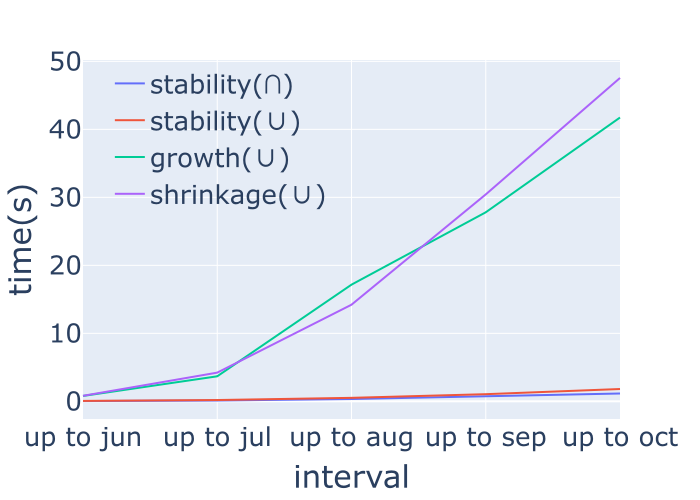}
\centering
\caption{}
\label{fig:mltime.2}%
\end{subfigure}
\begin{subfigure}[b]{0.33\textwidth}
\includegraphics[scale=0.21]{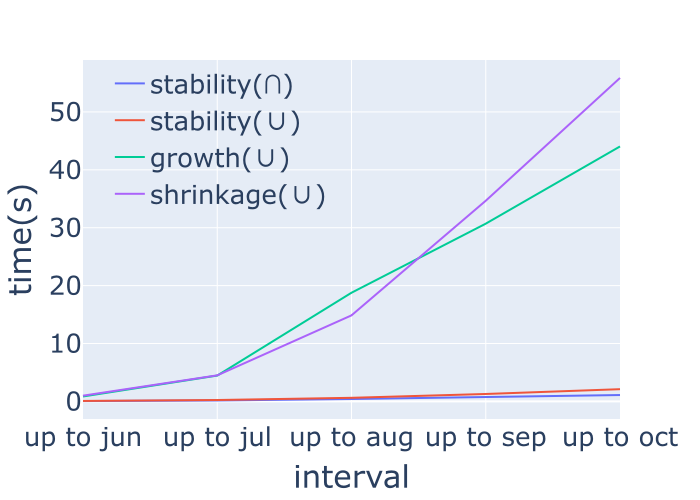}
\centering
\caption{}
\label{fig:mltime.3}%
\end{subfigure}
\caption{Skyline exploration execution time for \textit{MovieLens} for label co-rate and  aggregation by (a) gender, (b) age, (c) (gender, age).}
\label{fig:mltime}
\end{figure*}

\begin{figure*}
\centering
\begin{subfigure}[b]{0.33\textwidth}
\includegraphics[scale=0.21]{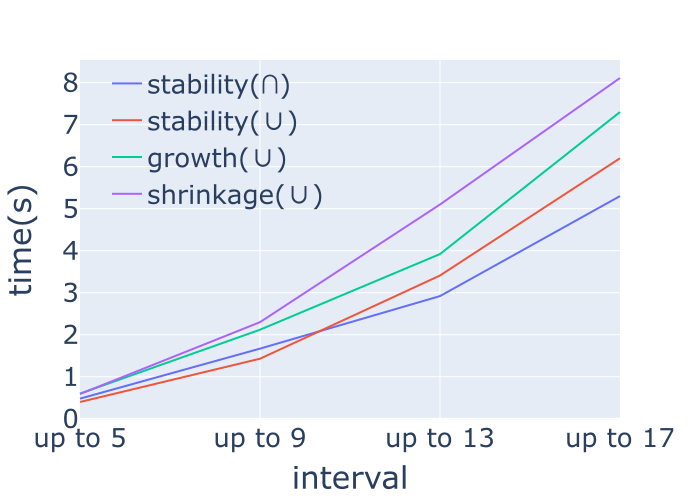}
\centering
\caption{}
\label{fig:pstime.1}
\end{subfigure}%
\begin{subfigure}[b]{0.33\textwidth}
\includegraphics[scale=0.21]{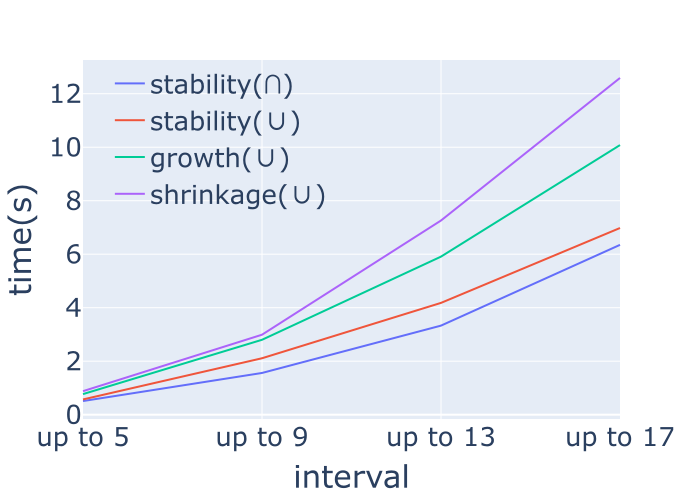}
\centering
\caption{}
\label{fig:pstime.2}%
\end{subfigure}
\begin{subfigure}[b]{0.33\textwidth}
\includegraphics[scale=0.21]{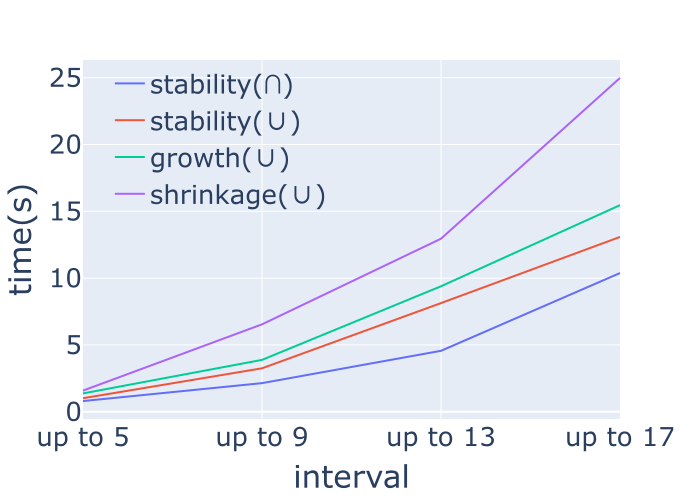}
\centering
\caption{}
\label{fig:pstime.3}%
\end{subfigure}
\caption{Skyline exploration execution time for \textit{Primary School} for label interact and aggregation by (a) gender, (b) class, (c) (gender, class).}
\label{fig:pstime}
\end{figure*}

\subsection{Qualitative evaluation}
In this set of experiments, we present qualitative results regarding the skyline-based exploration.

\paragraph{Skyline size.} First, let us look on the size of both individual and unified skylines.
In Table \ref{tab:skysize}, we report the size of the individual and the unified (denoted by *) skylines for $G_s(\cap)$, $G_s(\cup)$, $G_g(\cup)$, $G_h(\cup)$ for  \textit{DBLP}, \textit{MovieLens}, and \textit{Primary School}. Female and male values are represented as F and M, respectively. We notice that the skyline exploration that uses loose semantics generates larger skylines, probably because aggregation with loose semantics leads to larger graphs compared to the case of strict semantics. 

The size of the skyline also depends on the number of dimensions.
For example, the skyline for \textit{DBLP} and $G_s(\cap)$ has size 22 for the collaboration edges, where the number of dimensions is 5 (F-F, F-M, M-F, M-M, length), while for the publication edges the skyline size is 117 and the number of dimensions is 33. 
The number of time points also affects significantly the size of the skyline. 
For example, the sizes of the skylines for \textit{MovieLens} are in general  smaller than those for \textit{DBLP} even when the number of dimensions is large. For example,  for $G_g(\cup)$ and \textit{MovieLens} the skyline is of size 5 for co-rating edges and aggregation by gender with 5 dimensions, and of size 6 for aggregation by age with 50 dimensions. 

As a general observation, our findings are consistent with  Lemma \ref{lem4}, as we notice that the size for the unified skyline does not exceed the sum of the sizes of the individual skylines and also, it is larger than the size of any of the individual skylines.

Next, we present case studies of using skylines to explore graph evolution.

\paragraph{Individual skylines for stability.}


We zoom-in on individual skylines on gender property for \textit{Primary School}. The goal is to locate periods of stability for example to design mitigation strategies for disease containment. Figure \ref{fig:stab} presents the skyline output for the stable ($\cap$) interactions between girls (Fig. \ref{fig:stab.1}) and between boys (Fig. \ref{fig:stab.2}). 
There are longer stability intervals for boys than girls interactions. We observe that in the skylines for both girls and boys, intervals with length longer than 7 maintain at most 9 stable interactions, which can be seen as a bound on the duration and size of isolation bubbles. On the other hand, we observe an interval of duration 4 (interval [7, 11] at point of reference 12) with stability at least 32 for both girls and boys indicating a potential lower risk zone for disease spread.

\paragraph{Unified skylines for locating important time periods.}
After studying the individual skylines, we will examine the unified  skylines so that we can decide what is the most important period for each event. Since the unified skyline provides a high number of results, we will focus on the top-3 skylines based on \textit{dod}. In Table \ref{tab:topk}, we present the top-$3$ skyline results and the reported \textit{dod}, for the unified skyline exploration for each of our graph datasets and the events of $G_s(\cap)$, $G_s(\cup)$, $G_g(\cup)$, and $G_h(\cup)$. 

For \textit{DBLP}, we notice that the most important year for collaborations is 2020 where we notice high stability, growth, and shrinkage at the same time.
For \textit{MovieLens}, we notice that aug is the month with the greatest growth, as there is a significant increase in the number of edges for August (Table \ref{tab:ml}). For \textit{Primary School}, we notice that the 12th hour is the reference point with the highest stability ($G_s(\cap)$) for the aggregation by gender for the intervals [10, 11], [8, 11], and [7, 11], which seems reasonable as from 7th to 12th hour lessons are held and interactions between students are more stable comparing to breaks. Also, the 13th time point is the most important one for the event of growth for the aggregation by gender, as the 13th hour corresponds to break time where students interact with students from other classes as well.

\begin{table}
\caption{Skyline and unified skyline size}
\label{tab:skysize}
\centering
\begin{tabular}{c|ccccccc}
\hline
{} & {Aggregate} & {Edge type} & {$p-p'$} & {$G_s$($\cap$)} & {$G_s$($\cup$)} & {$G_g$($\cup$)} & {$G_h$($\cup$)}\\ \hline
\multirow{6}{*}{DBLP} & \multirow{5}{*}{gender} & \multirow{5}{*}{collaborate} & {F-F} & {$5$} & {$7$} & {$8$} & {$20$}\\
{} & {} & {} & {F-M} & {$8$} & {$13$} & {$14$} & {$20$}\\
{} & {} & {} & {M-F} & {$8$} & {$13$} & {$16$} & {$20$}\\
{} & {} & {} & {M-M} & {$12$} & {$19$} & {$19$} & {$20$}\\
{} & {} & {} & {*} & {$22$} & {$27$} & {$35$} & {$20$}\\
{} & {(gender, topic)} & {publish} & {*} & {$117$} & {$145$} & {$182$} & {$130$}\\ \hline
\multirow{7}{*}{ML} & \multirow{5}{*}{gender} & \multirow{7}{*}{co-rate} & {F-F} & {$3$} & {$3$} & {$5$} & {$5$}\\
{} & {} & {} & {F-M} & {$2$} & {$4$} & {$5$} & {$5$}\\
{} & {} & {} & {M-F} & {$2$} & {$4$} & {$5$} & {$5$}\\
{} & {} & {} & {M-M} & {$2$} & {$4$} & {$5$} & {$5$}\\
{} & {} & {} & {*} & {$3$} & {$4$} & {$5$} & {$5$}\\
{} & {age} & {} & {*} & {$10$} & {$15$} & {$6$} & {$7$}\\
{} & {(gender, age)} & {} & {*} & {$10$} & {$15$} & {$8$} & {$11$}\\ \hline
\multirow{6}{*}{PM} & \multirow{5}{*}{gender} & \multirow{6}{*}{interact} & {F-F} & {$10$} & {$9$} & {$12$} & {$13$}\\
{} & {} & {} & {F-M} & {$17$} & {$10$} & {$15$} & {$15$}\\
{} & {} & {} & {M-M} & {$13$} & {$9$} & {$13$} & {$15$}\\
{} & {} & {} & {*} & {$28$} & {$13$} & {$19$} & {$31$}\\
{} & {class} & {} & {*} & {$81$} & {$131$} & {$132$} & {$135$}\\
{} & {(gender, class)} & {} & {*} & {$103$} & {$130$} & {$131$} & {$136$}\\ \hline
\end{tabular}
\end{table}

\begin{figure*}[t]
\centering
\begin{subfigure}{0.49\textwidth}
\includegraphics[width=\textwidth]{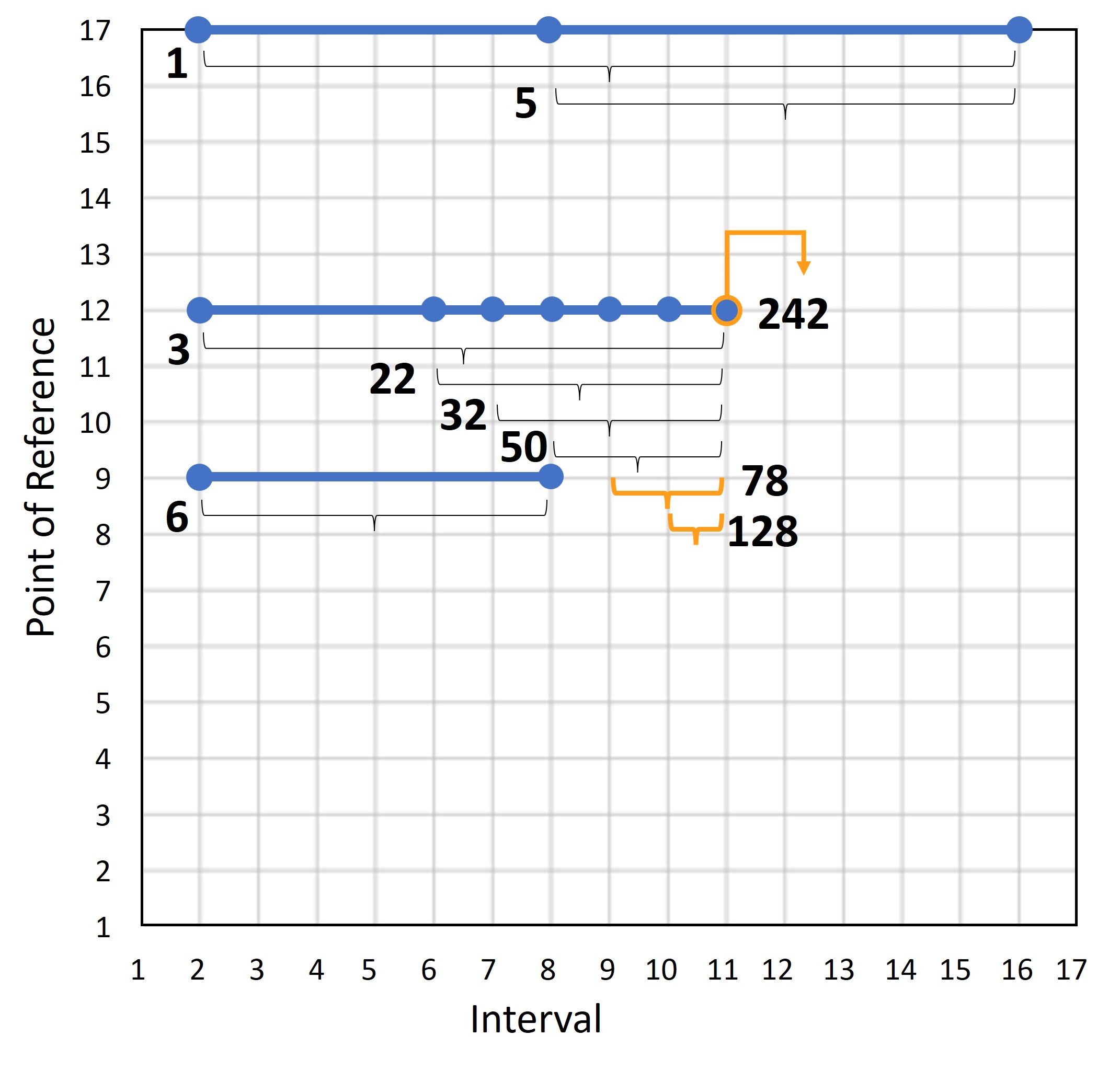}
\caption{Girls interactions}
\label{fig:stab.1}
\end{subfigure}
\begin{subfigure}{0.49\textwidth}
\includegraphics[width=\textwidth]{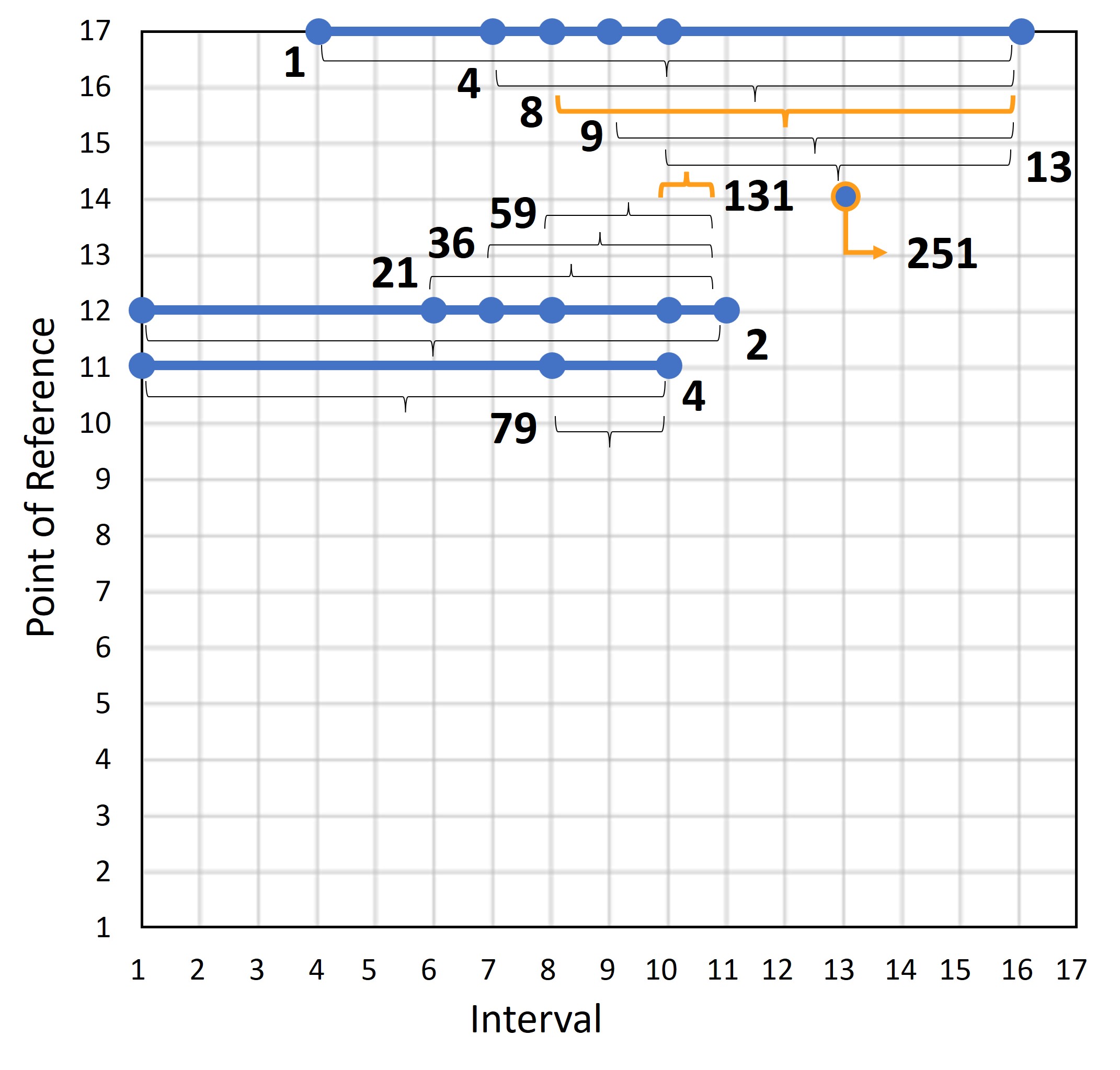}
\caption{Boys interactions}
\label{fig:stab.2}
\end{subfigure}
\caption{Skyline for the stability event ($\cap$) for (a) F-F and (b) M-M interactions for \textit{Primary School}.}
\label{fig:stab}
\end{figure*}

\begin{table}
\caption{Top-$3$ unified skyline}
\label{tab:topk}
\centering
\begin{tabular}{c|cc|cc}
\hline
\multicolumn{5}{c}{DBLP}\\ \hline
{} & {collaborate (gender)} &{dod} & {publish ((gender, topic))} & {dod}\\ \hline
\multirow{3}{*}{$G_s$($\cap$)} & {([2017], [2000,2016], 1)} & {44} & {([2018], [2008, 2017], 37)} & {8}\\
{} & {([2020], [2019], 2825)} & {19} & {([2017], [2008, 2016], 43)} & {6}\\
{} & {([2019], [2017, 2018], 645)} & {16} & {([2018], [2010, 2017], 67)} & {5}\\ \hline
\multirow{3}{*}{$G_s$($\cup$)} & {([2020], [2019], 2825)} & {19} & {([2018], [2016, 2017], 3467)} & {4}\\
{} & {([2020], [2018, 2019], 3688)} & {18} & {([2020], [2014, 2019], 6118)} & {3}\\
{} & {([2020], [2017, 2019], 3992} & {17} & {([2020], [2015, 2019], 5952)} & {3}\\ \hline
\multirow{3}{*}{$G_g$($\cup$)} & {([2019], [2000, 2018], 23475)} & {171} & {([2016], [2000, 2015], 6379)} & {7}\\
{} & {([2019], [2003, 2018], 23478)} & {2} & {([2018], [2013, 2017], 8941)} & {3}\\
{} & {([2020], [2001, 2019], 24117)} & {1} & {([2019], [2000, 2018], 9298)} & {2}\\ \hline
\multirow{3}{*}{$G_h$($\cup$)} & {([2020], [2019], 24630)} & {19} & {([2020], [2004, 2019], 82294)} & {4}\\
{} & {([2020], [2018, 2019], 42464)} & {18} & {([2019], [2004, 2018], 73725)} & {4}\\
{} & {([2020], [2017, 2019], 57659)} & {17} & {([2018], [2004, 2017], 66141)} & {4}\\ \hline
\multicolumn{5}{c}{ML}\\ \hline
{} & {co-rate (gender)} & {dod} & {co-rate (age)} & {dod}\\ \hline
\multirow{3}{*}{$G_s$($\cap$)} & {([sep], [aug], 957)} & {4} & {([sep], [aug], 957)} & {0}\\
{} & {([oct], [aug, sep], 30)} & {3} & {([aug], [jul], 573)} & {0}\\
{} & {([aug], [jul], 573)} & {0} & {([jul], [jun], 425)} & {0}\\ \hline
\multirow{3}{*}{$G_s$($\cup$)} & {([sep], [jul, aug], 1098)} & {6} & {([sep], [may, aug], 1185)} & {0}\\
{} & {([sep], [jun, aug], 1176)} & {4} & {([sep], [jun, aug], 1176)} & {0}\\
{} & {([sep], [aug], 957)} & {4} & {([sep], [jul, aug], 1098)} & {0}\\ \hline
\multirow{3}{*}{$G_g$($\cup$)} & {([aug], [may, jul], 609300)} & {9} & {([aug], [may, jul], 609300)} & {9}\\
{} & {([aug], [jun, jul], 609334)} & {4} & {([aug], [jun, jul], 609334)} & {4}\\
{} & {([sep], [may, aug], 76031)} & {4} & {([aug], [jul], 609477)} & {2}\\ \hline
\multirow{3}{*}{$G_h$($\cup$)} & {([sep], [aug], 609093)} & {4} & {([sep], [aug], 609093)} & {4}\\
{} & {([sep], [jul, aug], 810179)} & {3} & {([sep], [jul, aug], 810179)} & {3}\\
{} & {([sep], [jun, aug], 894867)} & {2} & {([sep], [jun, aug], 894867)} & {1}\\ \hline
\multicolumn{5}{c}{PS}\\ \hline
{} & {interact (gender)} & {dod} & {interact (class)} & {dod}\\ \hline
\multirow{3}{*}{$G_s$($\cap$)} & {([12], [10, 11], 513)} & {19} & {([11], [2, 10], 16)} & {9}\\
{} & {([12], [8, 11], 211)} & {19} & {([12], [2, 11], 12)} & {7}\\
{} & {([12], [7, 11], 138)} & {18} & {([12], [1, 11], 7)}& {6}\\ \hline
\multirow{3}{*}{$G_s$($\cup$)} & {([11], [6, 10], 1416)} & {79} & {([17], [2, 16], 1341)} & {1}\\
{} & {([11], [7, 10], 1383)} & {76} & {([16], [6, 15], 1000)} & {1}\\
{} & {([11], [8, 10], 1282)} & {72} & {([16], [8, 15], 993)} & {1}\\ \hline
\multirow{3}{*}{$G_g$($\cup$)} & {([13], [1, 12], 731)} & {134} & {([17], [1, 16], 356)} & {1}\\
{} & {([13], [2, 12], 732)} & {37} & {([14], [1, 13], 292)} & {1}\\
{} & {([13], [9, 12], 1112)} & {35} & {([16], [7, 15], 54)} & {1}\\ \hline
\multirow{3}{*}{$G_h$($\cup$)} & {([7], [2, 6], 3896)} & {28} & {([5], [2, 4], 3227)} & {1}\\
{} & {([5], [2, 4], 3121)} & {27} & {([16], [1, 15], 6101)} & {0}\\
{} & {([6], [2, 5], 3310)} & {21} & {([16], [2, 15], 6081)} & {0}\\ \hline
\end{tabular}
\end{table}

\paragraph{Comparing individual and unified skylines.}
After having studied the important periods as derived by the top-3 skylines for each event, we will examine the contribution of the individual skylines on the unified skyline. In other words, we will study the relevant importance of each individual skyline (i.e., each specific property value combination) in the results of the unified skyline. We will look into the most influencing property values for the top-3 skylines based on \textit{dod}.

In Fig. \ref{fig:topskydblp} and Fig. \ref{fig:topskyps}, we provide a visualization of the top-$3$ unified skylines for \textit{DBLP} and \textit{Primary school}, respectively.
To get more meaningful results, we have normalized the count of edges by dividing the count for each property value combination by the maximum reported count for this combination. Also, we have excluded from the results edges with zero count. Red-colored areas represent the 1st skyline result, green-colored areas the 2nd skyline result, and blue-colored areas the 3rd skyline result.

Figure \ref{fig:topskydblp.1} illustrates the top-$3$ results for the stable ($\cap$) publications for the aggregation on (gender, topic). We observe that (M, data engineering) and (M, pattern recognition) are the most important types of edges for the event of stability for all three results. Also, we  find only 4 results that contain  female authors, while 8 results for  male authors, which reflects that women are underrepresented in academic research through time. Regarding growth, in Fig. \ref{fig:topskydblp.2}, where we have excluded results with less than 80 edges (for presentation clarity), we observe that new publications in databases, data management, information retrieval and AI are the most significant ones regardless of gender, denoting the rapid growth of the research interest in those fields. 

Figure \ref{fig:topskyps} illustrates the top-$3$ skyline results for the event of stability ($\cap$) (Fig. \ref{fig:topskyps.1}) and growth ($\cup$) (Fig. \ref{fig:topskyps.2}) for the aggregation on class for \textit{Primary School}. In Fig. \ref{fig:topskyps.1} we observe that the reported results include only edges between individuals of the same class, which denotes that students of a class mainly interact and build stable contacts with students of the same class. The largest interval of stability is observed at the blue area for the 12th reference point and interval [1, 11] as reported in Table \ref{tab:topk}, which corresponds to class time. Also, the result in the red area is worth of further investigation as it achieves the highest counts on each of the reported edges. 
In Fig. \ref{fig:topskyps.2} for clarity of results, we have excluded results with less than 10 edges. The depicted results confirm our findings from Fig. \ref{fig:topskyps.1}, as we notice that new contacts occur mainly between students of different classes. Also, we notice that students of classes 1, 2 and 3 are the ones mainly participating in the results, implying that junior and mid-level students are more social compared to senior ones.

\begin{figure*}
\centering
\begin{subfigure}[b]{0.5\textwidth}
\includegraphics[scale=0.24]{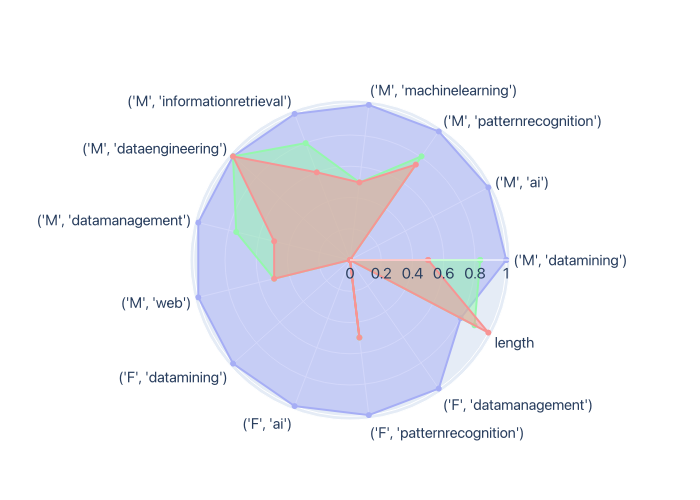}
\centering
\caption{}
\label{fig:topskydblp.1}
\end{subfigure}%
\begin{subfigure}[b]{0.5\textwidth}
\includegraphics[scale=0.24]{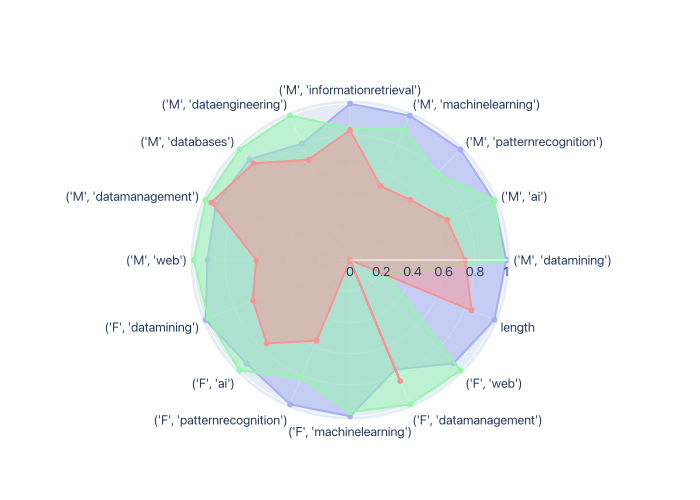}
\centering
\caption{}
\label{fig:topskydblp.2}%
\end{subfigure}
\caption{Top-$3$ skylines for (a) stable (stability ($\cap$)) and (b) new (growth ($\cup$)) publications for aggregation on (gender, topic) for \textit{DBLP}.}
\label{fig:topskydblp}
\end{figure*}

\begin{figure*}
\centering
\begin{subfigure}[b]{0.5\textwidth}
\includegraphics[scale=0.24]{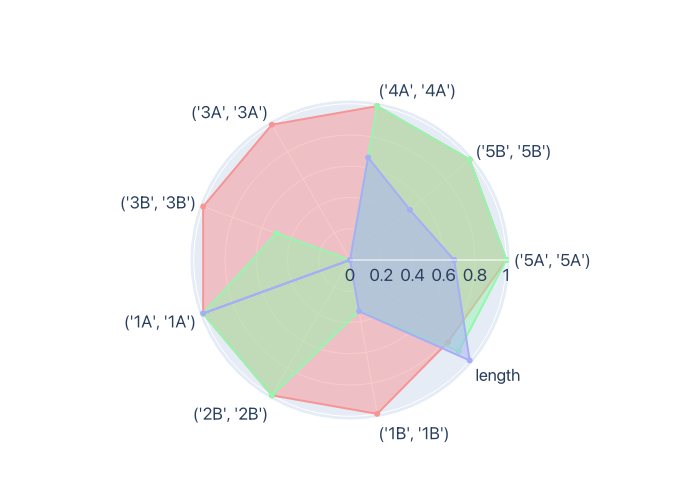}
\centering
\caption{}
\label{fig:topskyps.1}
\end{subfigure}%
\begin{subfigure}[b]{0.5\textwidth}
\includegraphics[scale=0.24]{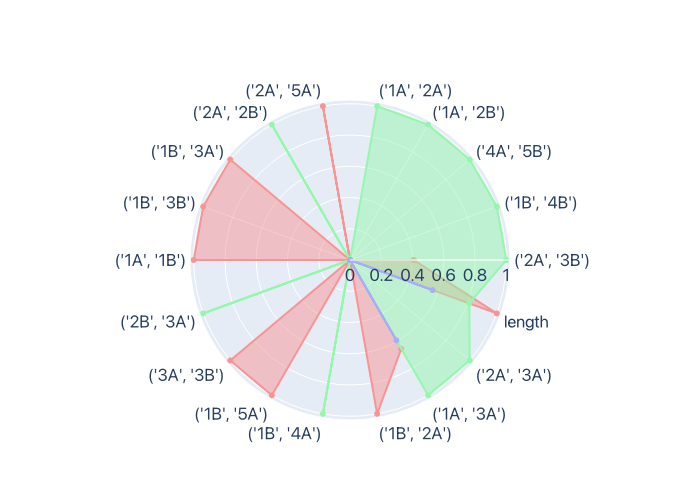}
\centering
\caption{}
\label{fig:topskyps.2}%
\end{subfigure}
\caption{Top-$3$ skylines for (a) stable (stability ($\cap$)) and (b) new (growth ($\cup$)) interactions for aggregation on class for \textit{Primary School}.}
\label{fig:topskyps}
\end{figure*}

\section{Related work}
In this paper, we have proposed exploration based on skylines, 
for identifying time periods that dominate other time periods in terms of increased activity (shrinkage, growth), or lack thereof (stability). To the best of our knowledge, this is a novel approach.



Regarding \textit{previous work on temporal graphs} and their evolution, there is some previous research in defining temporal operators and time and graph aggregation but provide no exploration strategies.
T-GQL \cite{Debrouvier21} is an extension of GQL 
with temporal operators to handle temporal paths; evolution is modeled through continuous and consecutive paths. In \cite{Orlando23} authors introduce TGV, a visualization query-driven tool based on T-GQL. The tool provides graph visualization for the type of paths presented in \cite{Debrouvier21}.
TGraph \cite{Moffitt17} uses temporal algebraic operators such as temporal selection for nodes and edges and traversal with temporal predicates on a temporal property graph. TGraph is extended with attribute and time aggregation that allows viewing a graph in different resolutions \cite{Aghasadeghi20} but only stability is studied.   GRADOOP \cite{Rost21,Rost22} introduces various extensions for supporting temporal property graphs such as temporal operators for grouping and pattern matching. The system provides different graph visualizations, i.e., the temporal graph view, the grouped graph view, and the difference graph view that illustrates new, stable and deleted elements between two graph snapshots. Unlike our work which facilitates a complete exploration strategy that concerns the history of the graph, the system is driven by user queries and provides no exploration strategy.

The GraphTempo framework \cite{Tsoukanara23} defines temporal operators and supports both time and attribute aggregation for temporal attributed graphs. An evolution graph is built by overlaying an intersection and two difference graphs and interesting events in intervals are detected based on a threshold-based strategy. An interactive tool for exploring the graph, its aggregations and the evolution graph is also provided \cite{TsoukanaraKP23a}. While our approach relies on the GraphTempo model,
here we focus on property graphs. 
Furthermore, the GraphTempo exploration strategy  requires selecting appropriate thresholds, our proposed skyline-based approach detects all events effectively without the need for parameter configuration.

Another approach to temporal graphs, that does not utilize temporal operators, is versioning (i.e., maintaining previous graph snapshots). In the EvOLAP graph \cite{Guminska18}, versioning is used both for attributes and graph structure to enable analytics on changing graphs. In \cite{Ghrab13}, explicit labeling of graph elements is designed to support analytical operations over evolving graphs, and particularly time-varying attributes, while in \cite{Andriamampianina22} the designed conceptual model captures changes on the topology, the set of attributes and the attributes’ values.


There has been a lot of \textit{previous work on skylines}. Since its introduction \cite{Borzsonyi01}, the skyline operator has been utilized in several domains to identify dominating entities in multi-criteria selection problems \cite{Kalyvas17}. 
Although skyline queries are very popular for multi-dimensional data, there is not much work on skylines  over graphs. 
A domain where skylines were  first used is road networks
 where the best detours based on a given route \cite{Huang04} or the best places to visit \cite{Jang08} are detected using distances among other possible criteria. In \cite{Kriegel10} skylines of routes based on multiple criteria, such as distance, and cost are also defined. The network is modeled as an multi-attribute graph and a vector of different optimization criteria is stored for each edge. In \cite{Chowdhury19}, authors explore the concept of skyline path queries in the context of location-based services, where given a pick-up point and a destination point the system applies skyline queries based on a set of features so as to pertain the most useful routes. 
 
Besides road networks,  skylines have been defined for graphs using the shortest path distances between nodes \cite{Zou10}. Specifically,
given a set of query nodes, a node $u$ dominates a node $v$, if $u$ is at least as close as $v$ to all query points and $u$ is closer than $v$ to at least one query point. A different approach that does not rely on  distances is defined in \cite{Weiguo16}, where a skyline consists of subgraphs that best match a given user query, also represented as a subgraph. Matching relies on isomorphisms and uses appropriate encoding schemes that capture both structural and numeric features of the graph nodes. Finally, skylines on knowledge graphs are defined in \cite{Keles19}. The focus is on supporting skyline queries over entities in an RDF graph through SPARQL queries and the efficient evaluation of such queries, but while the data are modeled as a graph, skylines are defined on node attributes and do not take into account graph structure.

Finally, this paper is an extended version of \cite{TsoukanaraADBIS23}.
While in \cite{TsoukanaraADBIS23}, we consider temporal attributed graphs,
in this paper, we focus on temporal property graphs and present a new form
of aggregation. Furthermore, we extend individual evolution skylines to unified evolution skylines and offer an extensive quantitative and qualitative evaluation.

\section{Conclusions}

We introduced a unified skyline-based exploration strategy to detect interesting intervals when a significant number of events occur in the evolution of a temporal property graph. We considered aggregate graphs and explored stability, growth and shrinkage, considering all value combinations for a given set of properties at the same time. We experimentally evaluated the efficiency and effectiveness of our approach and presented results on three real datasets. 

As a future work we plan to explore efficient strategies in limiting the number of property values taking into consideration the most influencing ones. We also plan to integrate our approach for graph aggregation and skyline-based exploration in an existing property graph database. Finally, since in many cases, properties have missing or erroneous values, another direction for future work is to extend aggregation to address this issue and also study its effect in the analysis of temporal evolution.


\section*{Declarations}

\subsection*{Ethics Approval and Consent to Participate}
Not applicable.

\subsection*{Consent for Publication}
Not applicable.

\subsection*{Availability of Data and Materials}
The datasets generated and/or analysed during the current study are available in the etsoukanara repository, \url{https://github.com/etsoukanara/uskylinexplore}.

\subsection*{Competing Interests}
The authors declare that they have no competing interests.

\subsection*{Funding}
Research work supported by the Hellenic Foundation for Research and Innovation (H.F.R.I.) under the “1st Call for H.F.R.I. Research Projects to Support Faculty Members \& Researchers and Procure High-Value Research Equipment” (Project Number: HFRI-FM17-1873, GraphTempo).

\subsection*{Authors' Contributions}
The study was designed by Evangelia Tsoukanara, Georgia Koloniari and Evaggelia Pitoura. Material preparation and data collection were performed by Evangelia Tsoukanara. The data analysis was performed by Evangelia Tsoukanara, and Georgia Koloniari and Evaggelia Pitoura provided additional expertise. The first draft of the manuscript was written by Evangelia Tsoukanara and all authors commented on previous versions of the manuscript. All authors read and approved the final manuscript.

\subsection*{Acknowledgements}
Not applicable.

%




\bibliography{bib}
\bibliographystyle{spmpsci}

\end{document}